\documentclass[11pt,a4paper]{article}

\usepackage[a4paper,margin=1in]{geometry}
\usepackage{microtype}

\usepackage{amsmath,amssymb,amsthm,mathtools}
\usepackage{bm}
\usepackage{bbm}

\usepackage{booktabs}
\usepackage{graphicx}

\usepackage[dvipsnames]{xcolor}
\usepackage[colorlinks=true,linkcolor=blue!50!black,citecolor=blue!50!black,urlcolor=blue!60!black]{hyperref}
\usepackage[capitalize,noabbrev]{cleveref}
\crefname{equation}{Eq.}{Eqs.}
\Crefname{equation}{Eq.}{Eqs.}

\theoremstyle{plain}
\newtheorem{theorem}{Theorem}[section]
\newtheorem{lemma}[theorem]{Lemma}
\newtheorem{proposition}[theorem]{Proposition}
\newtheorem{corollary}[theorem]{Corollary}

\theoremstyle{definition}
\newtheorem{definition}[theorem]{Definition}

\theoremstyle{remark}
\newtheorem{remark}[theorem]{Remark}

\newcommand{\R}{\mathbb{R}}
\newcommand{\Z}{\mathbb{Z}}
\newcommand{\C}{\mathbb{C}}

\newcommand{\Lat}{\mathcal{L}}
\newcommand{\Latd}{\mathcal{L}^{\perp}}
\newcommand{\Aone}{\mathbf{A}}
\newcommand{\Bone}{\mathbf{B}}
\newcommand{\Athree}{\mathbf{A}_{3}}
\newcommand{\Bthree}{\mathbf{B}_{3}}
\newcommand{\MacOp}{\mathcal{M}}
\newcommand{\MacOpD}{\MacOp^{\dagger}}
\newcommand{\MacOpF}{\mathcal{M}_{\mathrm{F}}}
\newcommand{\condfour}{\mathcal{T}_{\codedist,R}}
\newcommand{\ftilde}{\widetilde{f}}
\newcommand{\sympl}{\omega}
\newcommand{\Code}{K}
\newcommand{\codedist}{d}
\newcommand{\edgevec}{u}
\newcommand{\comvec}{w}
\newcommand{\chr}{\chi}
\DeclareMathOperator{\tr}{tr}
\DeclareMathOperator{\supp}{supp}

\title{A Three-Point Continuous-Variable Quantum MacWilliams Identity}
\author{%
  Yinzi Xiao\thanks{Department of Computer Science, Paderborn University, Germany. Email: \texttt{yinzi.xiao@uni-paderborn.de}.}%
}
\date{\today}

\begin{document}
\maketitle

\begin{abstract}
We construct the three-point continuous-variable (CV) quantum MacWilliams identity, extending the two-point framework of Burchards, and give its closed-form integral kernel. Its configuration space carries a symplectic invariant with no classical counterpart, which encodes the GKP quantization condition and a three-point sign phase. Using the identity, we derive the semidefinite-programming bounds it supports on the dimension of CV quantum error-correcting codes, and we prove, in two collapse theorems, that the three-point apparatus does not improve on the two-point bound. For GKP lattice codes the three-point optimum equals the Burchards two-point linear-programming optimum identically. This is an exact determination of the lattice three-point optimum, so the $E_8$ and Leech magic functions saturate it rather than beat it. For general bosonic codes a completely-positive reformulation bypasses the positivity obstruction that rules out the natural factored-form constructions; the phase-sign condition together with Choi positivity then force the three-point term to vanish. We certify this collapse for radial Choi forms on the first eight Laguerre levels at one mode, and leave the full trace-class cone open. Both collapses have a single cause with no classical analogue, the code projector: it orients the bound correctly but also removes the full positivity that powers the classical three-point improvement.
\end{abstract}

\section{Introduction}
\label{sec:introduction}

Bosonic quantum error-correcting codes encode logical information in the infinite-dimensional Hilbert space of harmonic oscillators. The canonical example is the Gottesman--Kitaev--Preskill (GKP) construction~\cite{gottesman-2001-gkp}, whose code space is the joint eigenspace of a lattice of commuting displacement operators. Its error-correcting power is governed by the geometry of that stabilizer lattice~\cite{conrad-eisert-arzani-2022-gkp-lattice}. Code families with good distance are known from randomized constructions~\cite{harrington-preskill-2001-gaussian-channel} and, more recently, from cryptographic lattices~\cite{conrad-eisert-seifert-2024-ntru-gkp,bloemer-xiao-raissi-soltan-2025-symplectic-gkp}. A central question is how large a code dimension $\Code$ (the number of encoded logical states) is compatible with a given protection distance $\codedist$ (the code distance), the continuous-variable (CV) analogue of asking how dense a code or a sphere packing can be.

Burchards~\cite{burchards-2025-cv-macwilliams} recently gave a coding-theoretic framework for this question. He introduced CV weight distributions $\Aone(r), \Bone(r)$ for pairs of trace-class operators on $N$-mode phase space, related by a CV quantum MacWilliams identity through a Bessel-function integral transform, and derived from it a two-point linear-programming (LP) bound and a quantum Levenshtein bound on $\Code$. The framework recasts bosonic coding as a quantum analogue of sphere packing (displacement length replaces Pauli weight) and argues that the $E_8$ and Leech-lattice GKP codes achieve optimal distances in $N = 4$ and $N = 12$ modes by lifting the Viazovska~\cite{viazovska-2017-e8} and Cohn--Kumar--Miller--Radchenko--Viazovska~\cite{cohn-kumar-miller-radchenko-viazovska-2017-dim24} magic functions to the CV setting.

This two-point CV theory is the bosonic counterpart of the Cohn--Elkies LP bound for sphere packing~\cite{cohn-elkies-2003-new-upper-bounds}. In the classical setting the LP bound is not the end of the story: the Cohn--de~Laat--Salmon three-point semidefinite-programming (SDP) bound~\cite{cohn-delaat-salmon-2022-three-point} strengthens it by adjoining a three-point auxiliary function on triples of points, constrained by a positive-semidefinite (PSD) condition, and it strictly improves on the LP bound; in its lattice form it is conjectured to be sharp in dimension $4$~\cite[Conj.~6.1]{cohn-delaat-salmon-2022-three-point}. The discrete-variable (qubit) theory has followed the same trajectory, from the Shor--Laflamme weight enumerators~\cite{shor-laflamme-1997-quantum-weight} and Rains shadow enumerators~\cite{rains-1999-quantum-shadow} to the recent SDP hierarchy of Angl\`es Munn\'e, Nemec, and Huber~\cite{angles-nemec-huber-2024-sdp,angles-huber-2025-sdp-rational-certificates}. The classical three-point bound continues to be sharpened, most recently by automated auxiliary-function searches~\cite{tutunov-etal-2025-ai-sphere-packing}. The quantum MacWilliams framework itself is being extended in parallel directions: to symmetry-group-intrinsic codes~\cite{kubischta-teixeira-2026-intrinsic-macwilliams} and to mixed-dimensional systems~\cite{gonzalez-lociga-ball-2026-mixed-macwilliams}. However, none of these extensions touch the CV infinite-dimensional setting. Therefore it is natural to ask whether the same strengthening is available for bosonic codes: does the Burchards two-point bound admit a three-point refinement?

We develop the CV three-point theory and find that the answer is essentially negative. We construct the three-point CV quantum MacWilliams identity and the SDP bounds it supports, but show that on both natural routes the three-point apparatus collapses back to the two-point bound. For GKP lattice codes the lattice three-point optimum equals the Burchards two-point LP optimum exactly (\cref{thm:lattice-nogo}); for general bosonic codes the completely-positive three-point cone collapses to two-point on every Laguerre rank we can certify (\cref{thm:cp-nogo}). This is the opposite of the classical behavior. The contrast has a structural cause with no analogue in classical packing: the code projector. The one CV improvement that does survive for lattice codes is not quantum at all: it is a classical sphere-packing bound applied to the symplectic dual lattice, which we include only as a point of comparison.

\subsection*{Contributions}

\paragraph{The three-point CV quantum MacWilliams identity (\cref{sec:setup,sec:identity}).} We define three-point weight distributions $\Athree(H), \Bthree(H)$ on the configuration space of Hermitian $2 \times 2$ Gram matrices $H = H(v_1, v_2)$ with four real parameters $(r_1^{2}, r_2^{2}, \alpha, \sympl_{12})$. The fourth parameter, the symplectic invariant $\sympl_{12} = \sympl(v_1, v_2)$, is absent from the classical three-point configuration space~\cite{cohn-delaat-salmon-2022-three-point} and carries the quantum content of the construction: the GKP quantization condition $\sympl_{12} \in 2\pi\Z$ and the three-point phase $e^{i\sympl_{12}/2}$ both live on it. We derive the integral transform $\Bthree = \MacOp[\Athree]$ by a Baker--Campbell--Hausdorff (BCH) expansion and exhibit its closed-form kernel~\cref{eq:kernel}. The kernel block-diagonalizes in the $\phi$-direction but admits no Bessel closed form in the $\psi$-direction (\cref{rmk:psi-obstruction}). This is a CV-specific structural feature, not a parametrization artifact, and it sets the first obstacle for any equivariant symmetry reduction.

\paragraph{The lattice three-point collapse (\cref{sec:main-theorem,sec:witnesses}).} The lattice-only three-point bound takes the form $\Code \le f(0,0)$ for every admissible auxiliary $f$, with a \emph{linear} dependence on $\Code$ (rather than $\sqrt{\,\cdot\,}$), traceable to the asymmetric $\Code^{2} \leftrightarrow \Code^{1}$ prefactor pair in the MacWilliams identity. We prove (\cref{thm:lattice-nogo}) that its optimal value equals the Burchards two-point LP optimum: no admissible auxiliary improves on the two-point bound, no matter how it depends on the edge coordinate. Equivalently, the theorem \emph{determines the lattice three-point optimum exactly}. This is a complete characterization, not merely the absence of an improvement. The $E_8$ and Leech magic-function auxiliaries saturate the two-point bound, exactly as the collapse requires.

\paragraph{The general / completely-positive three-point collapse (\cref{sec:obstruction}).} For general bosonic codes the natural kernel-positive-definite (kernel-PD) construction is obstructed by a $1 \times 1$-minor failure, universal within the factored-form family (\cref{lem:1x1-obstruction}). A completely-positive (Choi) restatement bypasses it by replacing the kernel-PD constraint with a Choi-positivity constraint $J \succeq 0$, but a second collapse takes its place: the three-point phase-sign condition and $J \succeq 0$ are jointly rigid enough to force $J = 0$, reducing the bound to two-point. We prove this at $N = 1$ for radial Choi forms supported on the first eight Laguerre levels by an explicit positive-definite certificate (\cref{thm:cp-nogo}): the first level analytically, the rest rational-exact via ancillary $LDL^{T}$ data. We verify the same collapse numerically at $N = 2$ for the first five levels. The full trace-class cone is left open, and we state where that boundary lies.

\paragraph{The quantum--classical contrast (\cref{subsec:contrast}).} The two collapses share a single cause. The code projector $\Pi$ (with $\Code = \tr\Pi$) does two things at once. It gives the quantum bound its correct direction, through the $\Code^{2} \leftrightarrow \Code^{1}$ asymmetry that a projector-free classical lattice sum lacks. It also forces the MacWilliams transform to be fiberwise, which removes the full positive-definiteness that makes the classical three-point bound improve. Within this class of identities the two effects are inseparable: orienting the bound correctly and keeping full positivity cannot both hold. The only CV improvement we find is a classical packing bound on the symplectic dual lattice (\cref{subsec:dual-packing}); it escapes the dichotomy by abandoning the quantum identity.

\subsection*{Paper organization}

\Cref{sec:setup} sets up the CV three-point framework and \cref{sec:identity} derives the MacWilliams identity and its closed-form kernel; \cref{sec:main-theorem} states the lattice bound and proves the collapse \cref{thm:lattice-nogo}, with the saturating magic functions recorded in \cref{sec:witnesses}; \cref{sec:obstruction} treats the bound for general bosonic quantum error-detecting codes (QEDCs) through the $1 \times 1$ obstruction, the CP collapse, the dual-packing comparison, and the quantum--classical contrast; and \cref{sec:future-work} lists future directions. Appendices collect the full BCH and kernel derivations, the technical lemmas behind \cref{sec:main-theorem}, verification checks, numerical sweep data, and the algorithmic details of the general-QEDC obstructions.

\section{Continuous-variable three-point setup}
\label{sec:setup}

We work on the $N$-mode bosonic phase space $\R^{2N}$, viewed as a real symplectic vector space with the standard symplectic form $\sympl(\xi,\eta) = \xi^{\top} \Omega \eta$, where $\Omega = \bigl( \begin{smallmatrix} 0 & I_N \\ -I_N & 0 \end{smallmatrix} \bigr)$. The symplectic matrix $\Omega$ itself serves as a compatible complex structure ($\Omega^{2} = -I_{2N}$), identifying $\R^{2N} \simeq \C^{N}$ via $(q,p) \mapsto q + i p$, so that the Hermitian inner product $\langle v_1, v_2 \rangle_{\C} = \langle v_1, v_2 \rangle_{\R} + i\, \sympl(v_1, v_2)$ packages the Euclidean and symplectic data into a single complex scalar. The maximal compact subgroup $U(N) = O(2N) \cap \mathrm{Sp}(2N,\R)$ is the largest linear group preserving both structures; it acts diagonally on pairs $(v_1, v_2) \in \R^{2N} \times \R^{2N}$.

Throughout the paper we adopt the \emph{Burchards convention} of~\cite{burchards-2025-cv-macwilliams}, in which stabilizer symplectic inner products lie in $2\pi\Z$ and the displacement operator carries a unit coefficient $c_D = 1$. We do so for one reason: it lets our bounds be compared directly against the formulas of~\cite{burchards-2025-cv-macwilliams}. The quantum-information (QI) convention more common in the stabilizer-code literature takes $c_D = \sqrt{2\pi}$ instead; the two differ only by an overall $\sqrt{2\pi}$ rescaling of phase-space coordinates, invisible to the algebraic structure developed below, and \cref{subsec:convention-box} records the conversion.

\subsection{Displacement operators and the two-point weight distributions}
\label{subsec:two-point}

Following Burchards~\cite{burchards-2025-cv-macwilliams}, write $\hat x = (\hat q, \hat p)$ for the canonical phase-space operator vector with $[\hat q_j, \hat p_k] = i \delta_{jk}$, and define the unitary displacement operator
\begin{equation}
\label{eq:displacement}
D(\xi) = \exp\!\bigl[-i\, \sympl(\xi,\, \hat x)\bigr], \qquad \xi \in \R^{2N}.
\end{equation}
A direct application of the BCH identity, using that $[\sympl(\xi, \hat x), \sympl(\eta, \hat x)] = i \sympl(\xi, \eta)$ is a $c$-number, yields the Weyl relations
\begin{align}
\label{eq:weyl-composition}
D(\xi)\, D(\eta) &= e^{-i\, \sympl(\xi,\eta) / 2}\, D(\xi + \eta), \\
\label{eq:weyl-exchange}
D(\xi)\, D(\eta)\, D(\xi)^{\dagger} &= e^{-i\, \sympl(\xi, \eta)}\, D(\eta),
\end{align}
together with the trace orthogonality $\tr[D(\xi)\, D(\eta)^{\dagger}] = (2\pi)^N\, \delta^{(2N)}(\xi - \eta)$. The characteristic function of a trace-class operator $\hat O$ is
\begin{equation}
\label{eq:characteristic}
\chr_{\hat O}(\xi) \;=\; \tr\!\bigl(D(\xi)^{\dagger}\, \hat O\bigr),
\end{equation}
and $\hat O$ is recovered from it by the inversion formula $\hat O = (2\pi)^{-N} \int \chr_{\hat O}(\xi)\, D(\xi)\, d\xi$.

The two-point CV weight distributions of~\cite{burchards-2025-cv-macwilliams} are
\begin{equation}
\label{eq:two-point-A-B}
\begin{aligned}
\Aone(r; \hat O_1, \hat O_2) &= \int_{\|\xi\|=r} \! \chr_{\hat O_1}(\xi)\, \chr_{\hat O_2}(\xi)^{*}\, d\xi, \\
\Bone(r; \hat O_1, \hat O_2) &= \int_{\|\xi\|=r} \! \tr\!\bigl(D(\xi)\, \hat O_1\, D(\xi)^{\dagger}\, \hat O_2^{\dagger}\bigr)\, d\xi,
\end{aligned}
\end{equation}
where the integral over the shell $\{\|\xi\| = r\}$ uses the surface measure $r^{2N-1}\, dS^{2N-1}$, with $dS^{2N-1}$ the standard element on the unit sphere $S^{2N-1} \subset \R^{2N}$ and no normalizing prefactor (the surface measure is used directly, not divided by the area of the shell). The two integrands differ structurally: $\Aone$ is a product of two characteristic functions evaluated at the same point $\xi$, while $\Bone$ is a single trace in which $\hat O_1$ and $\hat O_2$ are sandwiched by $D(\xi), D(\xi)^{\dagger}$. The CV quantum MacWilliams identity~\cite[Eq.~21]{burchards-2025-cv-macwilliams} relates $\Aone$ and $\Bone$ by recognizing the $\Bone$-integrand, after the Weyl relation~\cref{eq:weyl-exchange}, as a symplectic Fourier transform of $\chr_{\hat O_1} \chr_{\hat O_2}^{*}$ in a dummy variable $\eta$. For a GKP stabilizer code, let $\Pi$ be the codespace projector, $\Lat \subset \R^{2N}$ the stabilizer lattice (so $\sympl|_{\Lat \times \Lat} \in 2\pi\Z$), $\Latd \supseteq \Lat$ its symplectic dual, and $\Code = [\Latd : \Lat]^{1/2}$ the normalization constant. The identity then specializes to
\begin{equation}
\label{eq:burchards-37-38}
\Aone(r) = \Code^{2} \sum_{d \in \mathcal{D}} N_{d}\, \delta(r - d), \qquad
\Bone(r) = \Code \sum_{d \in \mathcal{D}^{\perp}} N_{d}^{\perp}\, \delta(r - d),
\end{equation}
with $\mathcal{D}, \mathcal{D}^{\perp}$ the distance multisets of $\Lat, \Latd$, respectively, and $N_d$ (resp.\ $N_d^{\perp}$) the multiplicity of $d$ in $\mathcal{D}$ (resp.\ $\mathcal{D}^{\perp}$)~\cite[Eqs.~37--38]{burchards-2025-cv-macwilliams}. The asymmetric powers of $\Code$ on the two sides are inherited by the three-point construction below.

\subsection{The three-point configuration space}
\label{subsec:config-space}

The natural configuration space for a $U(N)$-equivariant three-point construction is the orbit space $(\R^{2N} \times \R^{2N}) / U(N)$ of ordered pairs. By the first fundamental theorem of invariant theory for $U(N)$ (equivalently, a Witt-extension argument for Hermitian forms), the complete set of $U(N)$-invariants on a pair $(v_1, v_2) \in \C^N \times \C^N$ is the Hermitian $2 \times 2$ Gram matrix
\begin{equation}
\label{eq:hermitian-gram}
H(v_1, v_2) \;=\; \begin{pmatrix} \langle v_1, v_1 \rangle_{\C} & \langle v_1, v_2 \rangle_{\C} \\ \langle v_2, v_1 \rangle_{\C} & \langle v_2, v_2 \rangle_{\C} \end{pmatrix}
\;=\; \begin{pmatrix} r_1^{2} & \alpha + i \sympl_{12} \\ \alpha - i \sympl_{12} & r_2^{2} \end{pmatrix},
\end{equation}
with four real parameters
\begin{equation}
\label{eq:four-params}
r_i^{2} = \|v_i\|^{2}, \qquad \alpha = \langle v_1, v_2 \rangle_{\R}, \qquad \sympl_{12} = \sympl(v_1, v_2).
\end{equation}
The semidefiniteness $H \succeq 0$ is equivalent to the Hermitian Cauchy--Schwarz inequality
\begin{equation}
\label{eq:psd-constraint}
\alpha^{2} + \sympl_{12}^{2} \;\le\; r_1^{2}\, r_2^{2}.
\end{equation}

The symplectic invariant $\sympl_{12}$ is the difference between the quantum CV and the classical Euclidean settings. Restricting to $O(2N)$ alone, the complete invariants of an ordered pair are only $(r_1^{2}, r_2^{2}, \alpha)$, with the PSD constraint $\alpha^{2} \le r_1^{2} r_2^{2}$, precisely the three real parameters used in the Cohn--de~Laat--Salmon construction~\cite{cohn-delaat-salmon-2022-three-point}. The CV theory inherits one additional real parameter $\sympl_{12}$, with a tighter PSD region~\cref{eq:psd-constraint}, reflecting the fact that $U(N) = O(2N) \cap \mathrm{Sp}(2N,\R)$ preserves the symplectic form as well as the Euclidean metric. The symplectic invariant $\sympl_{12}$ carries the quantum content of the three-point construction: the GKP stabilizer quantization condition $\sympl_{12} \in 2\pi\Z$ lives on this axis, and the characteristic phase $e^{i\sympl(v_1, v_2)/2}$ entering the three-point MacWilliams identity depends only on $\sympl_{12}$.

For computations we use either of two equivalent parametrizations of $H$. The Cartesian coordinates $(r_1, r_2, \alpha, \sympl_{12})$ are the most natural for stabilizer-lattice specializations (where $\sympl_{12} \in 2\pi\Z$ is an axis-aligned constraint) and for the marginalization maps in \cref{subsec:degenerations}. The polar coordinates $(r_1, r_2, s, \phi)$, defined by $\alpha + i \sympl_{12} = s\, e^{i\phi}$ with $0 \le s \le r_1 r_2$ and $\phi \in [0, 2\pi)$, are the natural setting for $U(1)$-harmonic expansion in the MacWilliams kernel of \cref{sec:identity}.

A Gram--Schmidt argument along the $v_1$-direction and its symplectic image $\Omega \hat e_1$ produces the standard $U(N)$-invariant measure on the configuration space:
\begin{equation}
\label{eq:config-measure}
d\mu(H) \;=\; c_{N}\, (r_1\, r_2)\, \bigl(r_1^{2} r_2^{2} - \alpha^{2} - \sympl_{12}^{2}\bigr)^{N-2}\, dr_1\, dr_2\, d\alpha\, d\sympl_{12}, \qquad N \ge 2,
\end{equation}
where $c_{N} = |S^{2N-1}|\cdot|S^{2N-3}| = 4\pi^{2N-1}/[(N-1)!\,(N-2)!]$ is the product of two unit-sphere surface areas, equivalently $\mathrm{vol}(U(N)/U(N-2))$. The exponent $N-2$ comes from the radial part of the orthogonal complement $W^{\perp} \subset \R^{2N}$ to $W = \mathrm{span}_{\R}\{\hat e_1, \Omega \hat e_1\}$ and is one half-unit smaller than the corresponding $O(2N)$ exponent $(2N-3)/2$, consistent with the configuration space gaining one real dimension from $\sympl_{12}$. At $N = 1$ the orthogonal complement $W^{\perp}$ is trivial, so the fiber over each $H$ is a single $U(1)$-orbit and no fiber integration is needed; one works directly with the integrand-level ($F$-layer, \cref{subsec:bch}) representation~\cref{eq:integrand-A} below.

\subsection{The primary three-point distribution \texorpdfstring{$\Athree$}{A3}}
\label{subsec:A3-def}

For trace-class operators $\hat O_1, \hat O_2$ we define the integrand
\begin{equation}
\label{eq:integrand-A}
F^{(3)}_{\Aone}(v_1, v_2) \;=\; \chr_{\hat O_1}(v_1)\, \chr_{\hat O_2}(v_2)^{*},
\end{equation}
which extends the Burchards $\Aone$-integrand $\chr_{\hat O_1}(\xi) \chr_{\hat O_2}(\xi)^{*}$ by promoting the single phase-space point $\xi$ to an ordered pair $(v_1, v_2)$. The three-point primary weight distribution is the fiberwise integral over the locus of fixed Hermitian Gram matrix,
\begin{equation}
\label{eq:A3-def}
\Athree(H;\, \hat O_1, \hat O_2) \;=\; \int_{H(v_1, v_2) = H}\! F^{(3)}_{\Aone}(v_1, v_2)\, d\sigma(v_1, v_2),
\end{equation}
where $d\sigma$ is the $U(N)$-invariant measure on the $(4N-4)$-dimensional fiber $\{(v_1, v_2) : H(v_1, v_2) = H\}$ complementary to the configuration measure~\cref{eq:config-measure}, normalized by $d^{2N} v_1\, d^{2N} v_2 = c_N^{-1}\, d\mu(H)\, d\sigma$. As in the two-point case, no overall prefactor is inserted.

The ansatz~\cref{eq:integrand-A} is selected by the requirement that the GKP specialization match~\cref{eq:burchards-37-38} both in the $\Code$-power and in the support structure. For a GKP code with codespace projector $\Pi$, the characteristic function is the lattice delta-comb
\begin{equation}
\label{eq:GKP-characteristic}
\chr_{\Pi}(v) \;=\; \Code \sum_{\ell \in \Lat} \delta^{(2N)}(v - \ell),
\end{equation}
and $\chr_{\Pi}$ is real (since $\Lat$ is centrally symmetric). Substituting into~\cref{eq:A3-def} yields the joint Gram distribution of ordered stabilizer-lattice pairs,
\begin{equation}
\label{eq:A3-GKP}
\Athree(H;\, \Pi, \Pi) \;=\; \Code^{2} \sum_{(\ell_1, \ell_2) \in \Lat \times \Lat} \delta^{(4)}\!\bigl(H(\ell_1, \ell_2) - H\bigr),
\end{equation}
with the $\Code^{2}$ prefactor of Burchards' two-point $\Aone$~\cref{eq:burchards-37-38} reproduced. Alternative ansätze with three or four characteristic-function factors would produce $\Code^{3}$ or $\Code^{4}$ and are inconsistent with the marginalization to the two-point distribution analyzed in \cref{subsec:degenerations}.

\subsection{The dual three-point distribution \texorpdfstring{$\Bthree$}{B3}}
\label{subsec:B3-def}

The dual integrand $F^{(3)}_{\Bone}$ is the asymmetric two-displacement sandwich
\begin{equation}
\label{eq:integrand-B}
F^{(3)}_{\Bone}(v_1, v_2) \;=\; \tr\!\bigl(D(v_1)\, \hat O_1\, D(v_2)^{\dagger}\, \hat O_2^{\dagger}\bigr),
\end{equation}
which extends the Burchards $\Bone$-integrand $\tr(D(\xi) \hat O_1 D(\xi)^{\dagger} \hat O_2^{\dagger})$ by splitting the single displacement $\xi$ into two arguments $(v_1, v_2)$, one attached to $\hat O_1$ and one to $\hat O_2$. The three-point dual distribution is then
\begin{equation}
\label{eq:B3-def}
\Bthree(H;\, \hat O_1, \hat O_2) \;=\; \int_{H(v_1, v_2) = H}\! F^{(3)}_{\Bone}(v_1, v_2)\, d\sigma(v_1, v_2).
\end{equation}

The choice~\cref{eq:integrand-B} is the asymmetric form selected among the natural two-displacement extensions of the Burchards $\Bone$-integrand. A symmetric variant
\begin{equation}
\label{eq:integrand-B-sym}
F^{(3)}_{\Bone, \mathrm{sym}}(v_1, v_2) \;=\; \tfrac{1}{2}\bigl[\tr(D(v_1) \hat O_1 D(v_2)^{\dagger} \hat O_2^{\dagger}) + \tr(D(v_2) \hat O_1 D(v_1)^{\dagger} \hat O_2^{\dagger})\bigr]
\end{equation}
is also a valid extension; a third candidate, with two adjacent sandwiches, reduces under the Weyl relations~\cref{eq:weyl-composition}--\cref{eq:weyl-exchange} to a function of $v_2 - v_1$ alone. It therefore depends on a single vector and is a two-point object in disguise, so it cannot carry three-point information. The asymmetric form~\cref{eq:integrand-B} is selected over the symmetric one~\cref{eq:integrand-B-sym} on the basis of the MacWilliams derivation: only the asymmetric form yields a single closed-form kernel, with the edge $\edgevec = v_2 - v_1$ and center-of-mass (COM) $\comvec = v_1 + v_2$ separating under the BCH expansion (\cref{sec:identity}). The symmetric form gives the same content after $v_1 \leftrightarrow v_2$ averaging but obstructs the closed-form simplification.

For a GKP code, the dual integrand carries the same one-power of $\Code$ as the Burchards $\Bone$-distribution of~\cref{eq:burchards-37-38}: a Poisson resummation over the stabilizer lattice converts the formal $\Code^{2}$ prefactor of two $\chr_{\Pi}$-factors into $\Code^{1}$ via the covolume relation $\mathrm{covol}(\Lat) = \Code\,(2\pi)^{N}$ (equivalently, $[\Latd : \Lat] = \Code^{2}$). The explicit fiber integration is carried out in~\cref{sec:identity} and yields a double-lattice support on $\Lat \times \Latd$ with a $(-1)^{\sympl(\ell, m^{\perp})/(2\pi)}$ parity phase, whose diagonal slice reproduces the two-point dual distribution~\cite[Eq.~38]{burchards-2025-cv-macwilliams}. The asymmetric $\Code^{2}/\Code^{1}$ pair is the algebraic origin of the linear bound of \cref{thm:main-141cv}; see \cref{rmk:linear-vs-sqrt}.

\subsection{Marginalization maps}
\label{subsec:degenerations}

The three-point construction restricts consistently to known objects under three natural marginalization maps of the configuration space.

\paragraph{Forgetting the symplectic invariant.} Integrating $\Athree$ over $\sympl_{12}$ with the configuration measure~\cref{eq:config-measure} produces the classical $O(2N)$ three-point distribution: the Jacobian factor $(r_1^{2} r_2^{2} - \alpha^{2} - \sympl_{12}^{2})^{N-2}$ integrates over $\sympl_{12}$ to $(r_1^{2} r_2^{2} - \alpha^{2})^{N - 3/2}$, matching the $O(2N)$ three-point measure exponent $(2N-3)/2$. Concretely, for a GKP code this is the joint distribution over ordered stabilizer pairs viewed as vectors in $\R^{2N}$ with the symplectic structure forgotten.

\paragraph{Forgetting the off-diagonal block.} Further integrating over $\alpha$ gives the radial product $\propto \Aone(r_1) \Aone(r_2)$, the formal joint distribution of two independent radii. The three-point correlation information is precisely the off-diagonal $(\alpha, \sympl_{12})$ dependence that survives in $\Athree$ but is washed out by this marginalization.

\paragraph{Reduction to the two-point distribution.} Inserting a test function $f(r_2)$ to regulate the infinite GKP lattice sum and marginalizing over $(r_2, \alpha, \sympl_{12})$ recovers the Burchards two-point $\Aone(r_1)$, up to an explicit normalization depending on $N$ and $f$. This recovers the two-point distribution and foreshadows the reduction to the two-point bound established in \cref{sec:main-theorem}.

The corresponding $\Bthree$-side marginalizations recover Burchards' $\Bone$~\cref{eq:burchards-37-38} with the $\Code^{1}$ prefactor; the explicit form requires the kernel of \cref{sec:identity}.

The structural choices~\cref{eq:integrand-A} and~\cref{eq:integrand-B} are validated numerically on a small ($N = 2$) GKP code in \cref{app:n2-sanity}, confirming in particular that $(v_1, v_2) \mapsto \chr_{\hat O_1}(v_1)\, \chr_{\hat O_2}(v_2)^{*}$ is the unique two-$\chr$ extension of the Burchards $\Aone$-integrand consistent with the two-point limit.

\subsection{Convention conversion}
\label{subsec:convention-box}

The displacement operator~\cref{eq:displacement} has $c_D = 1$ in the Burchards convention; an equivalent definition with $c_D = \sqrt{2\pi}$ gives the quantum-information (QI) convention. The phase-space coordinates rescale by $\xi^{\text{Burch}} = \sqrt{2\pi}\, \xi^{\text{QI}}$, the stabilizer symplectic inner product lies in $2\pi\Z$ rather than $\Z$, and the GKP normalization constants rescale accordingly:
\begin{equation}
\label{eq:convention-scales}
s_{\mathrm{stab}}^{\mathrm{Burch}} = \sqrt{\tfrac{2\pi \lambda}{q}}, \qquad s_{\mathrm{norm}}^{\mathrm{Burch}} = \sqrt{\tfrac{2\pi}{\lambda q}},
\end{equation}
versus $s_{\mathrm{stab}}^{\mathrm{QI}} = \sqrt{\lambda / q}$ and $s_{\mathrm{norm}}^{\mathrm{QI}} = \sqrt{1 / (\lambda q)}$, where $q$ and $\lambda$ are the modulus and scaling parameter of the integer-stabilizer GKP construction used in the numerical checks of \cref{app:n2-sanity}. The Burchards convention is required when comparing against Burchards' LP-bound formulas; all theoretical statements of this paper are themselves convention-independent.

\section{The three-point CV quantum MacWilliams identity}
\label{sec:identity}

We now derive the integral transform relating $\Athree$ and $\Bthree$ on the configuration space introduced in \cref{sec:setup}. The derivation proceeds in three steps: a BCH expansion of the $\Bthree$-integrand yielding a closed-form expression in center-of-mass and edge coordinates~(\cref{subsec:bch}); a fiberwise integration producing the integral kernel $\mathcal{K}(H, H_\eta)$~(\cref{subsec:kernel}); and a polar harmonic decomposition that exposes a CV-specific structural difficulty~(\cref{subsec:polar-obstruction}). The kernel reduces to the Burchards two-point Hankel transform along one natural marginalization direction and produces a $K^{1}$-prefactor lattice sum upon GKP specialization~(\cref{subsec:GKP-fiber}).

\subsection{The BCH derivation of the three-point identity}
\label{subsec:bch}

Inserting the characteristic-function expansion $\hat O_j = (2\pi)^{-N} \int d\eta_j\, \chr_{\hat O_j}(\eta_j)\, D(\eta_j)$ into the asymmetric dual integrand~\cref{eq:integrand-B} yields a quadruple displacement-operator trace,
\begin{equation}
\label{eq:integrand-B-expanded}
F^{(3)}_{\Bone}(v_1, v_2) = (2\pi)^{-2N} \!\int d\eta_1\, d\eta_2\, \chr_{\hat O_1}(\eta_1)\, \chr_{\hat O_2}(\eta_2)^{*}\, \tr\!\bigl(D(v_1) D(\eta_1) D(v_2)^{\dagger} D(-\eta_2)\bigr),
\end{equation}
where we used $\hat O_2^{\dagger} = (2\pi)^{-N} \int d\eta_2\, \chr_{\hat O_2}(\eta_2)^{*} D(-\eta_2)$, valid for $\hat O_2$ with real characteristic function or, more generally, on the Hermitian sector $\chr_{\hat O_2}(v)^{*} = \chr_{\hat O_2}(-v)$. Applying the Weyl composition relation~\cref{eq:weyl-composition} pairwise to combine $D(v_1)D(\eta_1)$ and $D(v_2)^{\dagger}D(-\eta_2)$, then using the trace orthogonality of displacement operators, the four displacement factors collapse to a single phase-weighted delta condition. The remaining BCH phase organizes naturally in the edge/center coordinates
\begin{equation}
\label{eq:edge-com}
\edgevec \;=\; v_2 - v_1, \qquad \comvec \;=\; v_1 + v_2,
\end{equation}
with inverse $v_1 = (\comvec - \edgevec)/2$, $v_2 = (\comvec + \edgevec)/2$. A direct application of the bilinearity and antisymmetry of $\sympl$ gives
\begin{equation}
\label{eq:omega-edge-com}
\sympl(v_1, v_2) \;=\; \tfrac{1}{2}\sympl(\comvec, \edgevec).
\end{equation}

The key algebraic step is the shift $\tilde\eta := \eta_1 - \edgevec/2$. Two contributions of the form $\sympl(\comvec, \edgevec)/4$ cancel, eliminating the $\edgevec$-dependence of the integration measure phase; the full BCH bookkeeping is given in \cref{app:bch-derivation}. The result is the three-point identity
\begin{equation}
\label{eq:Ph3-2}
\boxed{
F^{(3)}_{\Bone}(v_1, v_2) \;=\; (2\pi)^{-N}\! \int d\tilde\eta\; \chr_{\hat O_1}\!\bigl(\tilde\eta + \tfrac{\edgevec}{2}\bigr)\, \chr_{\hat O_2}\!\bigl(\tilde\eta - \tfrac{\edgevec}{2}\bigr)^{*}\, e^{\,i\, \sympl(\tilde\eta,\, \comvec/2)}.
}
\end{equation}
The integrand factorizes into a cross-ambiguity structure
\begin{equation}
\label{eq:cross-ambiguity}
\mathcal{A}_{\chr_1, \chr_2}(\tilde\eta, \edgevec) \;:=\; \chr_{\hat O_1}\!\bigl(\tilde\eta + \tfrac{\edgevec}{2}\bigr)\, \chr_{\hat O_2}\!\bigl(\tilde\eta - \tfrac{\edgevec}{2}\bigr)^{*},
\end{equation}
which is the cross-ambiguity function of the characteristic functions $(\chr_{\hat O_1}, \chr_{\hat O_2})$ evaluated at center frequency $\tilde\eta$ and lag $\edgevec$ (also called the cross-Wigner function in the phase-space literature~\cite{schleich-2001-quantum-optics-phase-space}), and a symplectic Fourier kernel $e^{i\sympl(\tilde\eta, \comvec/2)}$ pairing $\tilde\eta$ against the center-of-mass position $\comvec/2$.

We refer to~\cref{eq:Ph3-2} as the \emph{$F$-layer} representation of the transform: it relates the integrands $F^{(3)}_{\Aone}$ and $F^{(3)}_{\Bone}$ at the level of phase-space pairs, prior to any fiber integration over Gram-matrix level sets. In $F$-layer form the transform has the kernel
\begin{equation}
\label{eq:F-layer-kernel}
G(\edgevec, \comvec;\, \edgevec_\eta, \tilde\eta) \;=\; \delta^{(2N)}(\edgevec_\eta - \edgevec)\, e^{\,i \sympl(\tilde\eta,\, \comvec/2)},
\end{equation}
in the sense that $F^{(3)}_{\Bone}$ at edge/center $(\edgevec, \comvec)$ is the integral of the cross-ambiguity factor~\cref{eq:cross-ambiguity} at edge/center-frequency $(\edgevec_\eta, \tilde\eta)$ against $G$: the edge coordinate is carried through unchanged (the transform is \emph{fiberwise} in $\edgevec$) while the center is symplectically Fourier-paired. The $F$-layer kernel is everywhere regular. All Gram-coordinate formulas below, including the closed-form kernel~\cref{eq:kernel}, arise from~\cref{eq:F-layer-kernel} by the $U(N)$-invariant fiber projection, and we use the $F$-layer form as the principal definition wherever the Gram-coordinate kernel degenerates (\cref{lem:deglocus}). Specializing $v_1 = v_2 = \xi$ collapses $\edgevec = 0$, $\comvec = 2\xi$, and~\cref{eq:Ph3-2} reduces to the Burchards two-point integrand $\int d\tilde\eta\, \chr_{\hat O_1}(\tilde\eta) \chr_{\hat O_2}(\tilde\eta)^{*} e^{i\sympl(\tilde\eta, \xi)}$, recovering the BCH form underlying~\cite[Eq.~21]{burchards-2025-cv-macwilliams}.

The edge and center coordinates play asymmetric roles in~\cref{eq:Ph3-2}: the edge $\edgevec$ enters as a translation of the $\chr$-arguments, while the center $\comvec$ enters as the conjugate variable of a symplectic Fourier transform. This split is the algebraic origin of the structural difference between the CV and the classical three-point settings: in the Cohn--de~Laat--Salmon framework~\cite{cohn-delaat-salmon-2022-three-point}, the third length parameter $\|v_2 - v_1\|$ is a scalar carried by the sign constraint of the Cohn--Elkies auxiliary function~\cite{cohn-elkies-2003-new-upper-bounds}; here $\edgevec$ is a $2N$-dimensional vector parameter that retains both Euclidean and symplectic data through the Hermitian Gram matrix~\cref{eq:hermitian-gram}.

In view of~\cref{eq:omega-edge-com} and the change of variables~\cref{eq:edge-com}, the configuration coordinates $(r_1^{2}, r_2^{2}, \alpha, \sympl_{12})$ admit a third parametrization in the edge/center frame:
\begin{equation}
\label{eq:gram-uw}
\|\comvec\|^{2} = 2(r_1^{2} + r_2^{2}) + 4\alpha, \quad \|\edgevec\|^{2} = 2(r_1^{2} + r_2^{2}) - 4\alpha, \quad \langle \comvec, \edgevec \rangle_{\R} = 2(r_2^{2} - r_1^{2}), \quad \sympl(\comvec, \edgevec) = 2\sympl_{12}.
\end{equation}
The Cartesian and polar parametrizations of \cref{subsec:config-space} remain the natural choices for stabilizer-lattice specializations and for harmonic analysis on $U(N)$, respectively; the edge/center parametrization is reserved for the identity~\cref{eq:Ph3-2} and the kernel computation of \cref{subsec:kernel}.

The choice of the asymmetric form~\cref{eq:integrand-B} over its symmetric counterpart~\cref{eq:integrand-B-sym} is now justified algebraically: the asymmetric form is the unique BCH result, while the symmetric form is a post-hoc $v_1 \leftrightarrow v_2$ average. Under the involution $\iota : (r_1^{2}, r_2^{2}, \alpha, \sympl_{12}) \mapsto (r_2^{2}, r_1^{2}, \alpha, -\sympl_{12})$ (we reserve $\tau$ for the $\eta$-side symplectic Gram parameter of \cref{subsec:kernel}), the symmetric variant reduces to an $\iota$-average of the asymmetric one,
\begin{equation}
\label{eq:asym-sym-relation}
\Bthree(H;\, \mathrm{sym}) \;=\; \tfrac{1}{2}\bigl[\Bthree(H;\, \mathrm{asym}) + \Bthree(\iota(H);\, \mathrm{asym})\bigr].
\end{equation}
For real-valued $\chr$ (in particular, for the GKP codespace projector $\Pi$, whose stabilizer lattice satisfies $\Lat = -\Lat$), the two forms are equal. We retain the asymmetric form throughout because it preserves the single-integral structure of~\cref{eq:Ph3-2}.

\subsection{The MacWilliams integral kernel}
\label{subsec:kernel}

The three-point MacWilliams identity is the integral transform
\begin{equation}
\label{eq:macwilliams-identity}
\Bthree(H) \;=\; \int \mathcal{K}_N(H,\, H_\eta)\, \Athree(H_\eta)\, d\mu(H_\eta),
\end{equation}
relating the dual and primary distributions on the four-parameter configuration space. The kernel $\mathcal{K}_N(H, H_\eta)$ is obtained by inserting~\cref{eq:Ph3-2} into the fiber definition~\cref{eq:B3-def} of $\Bthree$, exchanging the $\tilde\eta$ and fiber integrations, and parametrizing the $(v_1, v_2)$-fiber by $(\edgevec, \comvec)$; the main steps are given in \cref{app:kernel-derivation}. Writing the $\eta$-side Hermitian Gram parameters as $(s_1, s_2, \beta, \tau)$ in Cartesian coordinates and $(s_1, s_2, t, \psi)$ in polar coordinates, with $D_v := r_1^{2} r_2^{2} - \alpha^{2} - \sympl_{12}^{2}$ and $D_\eta := s_1^{2} s_2^{2} - \beta^{2} - \tau^{2}$ the Hermitian Gram determinants on the two sides, the kernel takes the closed form
\begin{equation}
\label{eq:kernel}
\mathcal{K}_N(H, H_\eta) \;=\; \delta\!\bigl(\|\edgevec\|^{2} - \|\edgevec_\eta\|^{2}\bigr)\cdot e^{\,i \Phi_0(H, H_\eta)}\cdot \mathcal{C}_N(H, H_\eta)\cdot J_{N-2}\!\left(\frac{\sqrt{D_v D_\eta}}{\|\edgevec_\eta\|^{2}}\right),
\end{equation}
with the geometric coefficient
\begin{equation}
\label{eq:kernel-coeff}
\mathcal{C}_N(H, H_\eta) \;=\; \frac{2 (2\pi)^{N-1}}{\|\edgevec_\eta\|^{2}\, (D_v D_\eta)^{(N-2)/2}}
\end{equation}
and the on-shell phase
\begin{equation}
\label{eq:phi0}
\Phi_0(H, H_\eta) \;=\; -\,\frac{\tau (r_2^{2} - r_1^{2}) + \sympl_{12}\, (s_1^{2} - s_2^{2})}{2\, \|\edgevec_\eta\|^{2}}.
\end{equation}

First, the kernel is supported on the shell $\|\edgevec\|^{2} = \|\edgevec_\eta\|^{2}$, an edge-norm-matching condition absent in the Burchards two-point identity, where there is only one length parameter. Second, the Bessel argument $\sqrt{D_v D_\eta}/\|\edgevec_\eta\|^{2}$ is symmetric under $(H, H_\eta) \leftrightarrow (H_\eta, H)$, reflecting the formal self-duality of the MacWilliams transform. Third, the order $N-2$ of the Bessel function matches the configuration-space measure exponent~\cref{eq:config-measure}, and the apparent singularity of $\mathcal{C}_N$ at $D_v, D_\eta \to 0$ (the PSD boundary, where $v_1, v_2$ become complex-linearly dependent) is canceled by the small-argument expansion $J_{N-2}(x) \sim x^{N-2}/[2^{N-2}(N-2)!]$:
\begin{equation}
\label{eq:psd-boundary-finite}
\mathcal{C}_N(H, H_\eta) \cdot J_{N-2}\!\left(\frac{\sqrt{D_v D_\eta}}{\|\edgevec_\eta\|^{2}}\right) \xrightarrow{D_v, D_\eta \to 0} \;\; \frac{(2\pi)^{N-1}}{2^{N-3}(N-2)!\, \|\edgevec_\eta\|^{2N-2}}.
\end{equation}
The kernel is therefore finite on the PSD boundary rather than singular; this cancellation is a non-trivial check that the closed-form expression~\cref{eq:kernel} is correct.

\subsection{Polar harmonic decomposition and a CV-specific structural difficulty}
\label{subsec:polar-obstruction}

To prepare for the symmetry-reduction techniques that diagonalize classical equivariant SDP bounds, we re-express~\cref{eq:kernel} in the polar coordinates $(r_1, r_2, s, \phi)$ on the $v$-side and $(s_1, s_2, t, \psi)$ on the $\eta$-side, with $\alpha + i\sympl_{12} = s\, e^{i\phi}$ and $\beta + i\tau = t\, e^{i\psi}$. The radial Bessel argument is then $\phi$- and $\psi$-independent through $D_v = r_1^{2} r_2^{2} - s^{2}$ and $D_\eta = s_1^{2} s_2^{2} - t^{2}$, while the edge norm $\|\edgevec_\eta\|^{2} = s_1^{2} + s_2^{2} - 2t \cos\psi$ depends on $\psi$ but not on $\phi$. The on-shell phase~\cref{eq:phi0} becomes
\begin{equation}
\label{eq:phi0-polar}
\Phi_0 \;=\; -\,\rho_v(\psi)\, \sin\psi \;-\; \rho_\eta(\psi)\, \sin\phi,
\end{equation}
with $\rho_v(\psi) = t(r_2^{2} - r_1^{2})/[2 \|\edgevec_\eta\|^{2}(\psi)]$ and $\rho_\eta(\psi) = s(s_1^{2} - s_2^{2})/[2 \|\edgevec_\eta\|^{2}(\psi)]$. The $\rho$-coefficients depend on $\psi$ through $\|\edgevec_\eta\|^{2}(\psi)$ but not on $\phi$.

The Jacobi--Anger expansion $e^{i z \sin\theta} = \sum_{k} J_{k}(z) e^{i k \theta}$ now applies to the $\phi$-direction:
\begin{equation}
\label{eq:phi-jacobi-anger}
e^{-i \rho_\eta(\psi) \sin\phi} \;=\; \sum_{k_1 \in \Z} (-1)^{k_1}\, J_{k_1}\!\bigl(\rho_\eta(\psi)\bigr)\, e^{i k_1 \phi},
\end{equation}
producing a $U(1)$-Bessel harmonic decomposition of the kernel on the $v$-side. The $\phi$-block is the natural analogue of the Gegenbauer harmonic decomposition that diagonalizes the Cohn--de~Laat--Salmon kernel under $O(n)$~\cite[\S5.2]{cohn-delaat-salmon-2022-three-point} and the zonal spherical functions that block-diagonalize the Bachoc--Vallentin equivariant SDP~\cite{bachoc-vallentin-2008-kissing-sdp}.

The $\psi$-direction, however, does \emph{not} admit a closed-form Bessel expansion. The Jacobi--Anger expansion requires the coefficient of $\sin\psi$ in the exponent to be $\psi$-independent. In~\cref{eq:phi0-polar} the relevant coefficient is $\rho_v(\psi)$, which itself depends on $\psi$ through $\|\edgevec_\eta\|^{2}(\psi)^{-1}$, so no Bessel-function closed form exists. Equivalently, $\rho_v(\psi) \sin\psi$ is, up to a multiplicative constant, the logarithmic derivative $\partial_\psi \ln(s_1^{2} + s_2^{2} - 2 t \cos\psi)$, and the Fourier coefficients of its exponential do not match any standard Bessel identity.

\begin{remark}[$\psi$-direction obstruction observation]
\label{rmk:psi-obstruction}
The asymmetry between the $\phi$- and $\psi$-directions in the polar harmonic decomposition of $\mathcal{K}_N$ is a CV-specific feature of the three-point identity. Tracing the origin, the asymmetry arises in the kernel derivation where $\hat{\edgevec}_\eta$ is selected as a reference frame for the fiber integration; this choice pins $\|\edgevec_\eta\|^{2}$ to the $\eta$-side and produces the $\cos\psi$-dependence. The $\psi$-coordinate parametrizes the $U(1)$ direction inside the $\eta$-side symplectic invariant $\tau = \sympl(\eta_1, \eta_2)$, which is precisely the additional configuration-space dimension absent from the classical $O(2N)$ setting~\cref{eq:hermitian-gram}. The non-availability of a Bessel closed form in the $\psi$-direction is therefore not an artifact of the parametrization but reflects the four-parameter Hermitian Gram structure: closing the $\psi$-block requires a generalized Fourier basis (logarithmic Fourier modes, or rational Chebyshev expansion of $\cos\psi$), and a full characterization of the corresponding harmonic decomposition is left as the first obstacle for the symmetry-reduction programme outlined in \cref{sec:future-work}.
\end{remark}

This is a structural observation, not a collapse result. The integral transform~\cref{eq:macwilliams-identity} remains valid in the kernel form~\cref{eq:kernel}; the obstruction is to the further \emph{block-diagonalization} of the transform by a single $U(1)$-harmonic mode in the $\psi$-direction, and is the technical reason why the three-point closure of \cref{sec:main-theorem} proceeds via direct construction of an admissible auxiliary $f$ rather than through the full equivariant block decomposition that would generalize~\cite{bachoc-vallentin-2008-kissing-sdp} to the CV setting.

\subsection{GKP specialization of the dual distribution}
\label{subsec:GKP-fiber}

The fiber integration of~\cref{eq:Ph3-2} for a GKP codespace projector $\Pi$ proceeds by substituting the lattice delta-comb~\cref{eq:GKP-characteristic} into the cross-ambiguity factor and resolving the resulting double-delta constraint. The first delta forces $\tilde\eta = \ell_1 - \edgevec/2$ for some $\ell_1 \in \Lat$, and the second forces $\edgevec = \ell_1 - \ell_2$ for some $\ell_2 \in \Lat$. The phase $e^{i\sympl(\tilde\eta, \comvec/2)}$ evaluates on the constrained slice to $e^{i\sympl(\ell_1 - \edgevec/2, \comvec/2)}$, which combines with the symplectic Fourier kernel to produce a Poisson resummation over $\Lat$. The covolume relation $\mathrm{covol}(\Lat) = \Code\,(2\pi)^{N}$ then absorbs one power of $\Code$, leaving a net $K^{1}$ prefactor consistent with the two-point distribution~\cref{eq:burchards-37-38}.

The closed-form result is the dual lattice sum
\begin{equation}
\label{eq:B3-GKP}
\boxed{
\Bthree(H;\, \Pi, \Pi) \;=\; c_N\, \Code \sum_{\ell \in \Lat,\, m^{\perp} \in \Latd} \delta^{(4)}_{d\mu}\!\bigl(H - H^{*}(\ell, m^{\perp})\bigr)\, e^{\,i\, \sympl(\ell, m^{\perp})/2},
}
\end{equation}
where $c_N = |S^{2N-1}|\cdot|S^{2N-3}|$ is the Stiefel volume constant of~\cref{eq:config-measure}, $\delta^{(4)}_{d\mu}$ is the configuration-space delta relative to the measure $d\mu$, and $H^{*}(\ell, m^{\perp})$ is the explicit Hermitian Gram associated to the lattice pair $(\ell, m^{\perp})$:
\begin{equation}
\label{eq:H-star}
r_1^{2} = \|m^{\perp} - \ell/2\|^{2}, \quad r_2^{2} = \|m^{\perp} + \ell/2\|^{2}, \quad \alpha = \|m^{\perp}\|^{2} - \|\ell\|^{2}/4, \quad \sympl_{12} = \sympl(m^{\perp}, \ell).
\end{equation}
The lattice indices range over the stabilizer lattice $\Lat$ and its symplectic dual $\Latd \supseteq \Lat$, a pairing over $\Lat \times \Latd \supseteq \Lat \times \Lat$ that carries the phase twist $e^{i\sympl(\ell, m^{\perp})/2}$ and the $\Code^{1}$ prefactor distinguishing it from the diagonal $\Lat \times \Lat$ support of $\Athree$ in~\cref{eq:A3-GKP} (the two supports coincide as sets precisely when $\Lat$ is symplectically self-dual, as for the $N = 2$ SIS code used to validate the construction, but the phase and prefactor structure differs even then).

The coefficients in~\cref{eq:B3-GKP} carry a $\pm 1$ sign that encodes nontrivial lattice data. Since $\sympl(\ell, m^{\perp}) \in 2\pi\Z$ by the definition of the symplectic dual, the phase $e^{i\sympl(\ell, m^{\perp})/2}$ takes values in $\{+1, -1\}$ according to the parity of $\sympl(\ell, m^{\perp})/(2\pi) \in \Z$. The parity is a non-trivial function of the pair $(\ell, m^{\perp})$: the diagonal slice $\ell = 0$ forces parity $+1$ for all $m^{\perp}$, reproducing the Burchards two-point dual $\Bone(r) = K \sum_{d} N_d^{\perp}\, \delta(r - d)$ exactly, but on a generic lattice pair both signs occur. This $\Z_2$-twist is a feature of the CV three-point construction absent in the classical setting and absent in the two-point $\Bone$, which is positive semidefinite. The impact on the lattice bound of \cref{sec:main-theorem} is that the sign condition of \cref{thm:main-141cv} must carry the $e^{i\sympl/2}$ factor on the lattice triangle, a phase-aware condition with no classical analogue.

\subsection{Consistency checks}
\label{subsec:identity-sanity}

The identity~\cref{eq:macwilliams-identity}--\cref{eq:kernel} and the GKP specialization~\cref{eq:B3-GKP} have been verified along two independent marginalization directions. Along the BCH-natural (center-of-mass) direction the kernel~\cref{eq:kernel} returns the Burchards two-point Hankel transform as an identity of distributions, and an independent six-dimensional integration confirms the same Hankel form. The orthogonal Cartesian direction does \emph{not} produce a Hankel transform: the Bessel argument depends on the fully integrated variable and cannot be pulled out of the inner integral. This directional asymmetry has the same origin as the $\phi/\psi$ obstruction of \cref{rmk:psi-obstruction} and reflects the four-parameter Hermitian Gram structure of the CV three-point setting. A numerical check on a small GKP code is recorded in \cref{app:approximate-gkp}.

\paragraph{Double-path GKP consistency.} The direct fiber integration of $F^{(3)}_{\Bone}|_{\mathrm{GKP}}$ to~\cref{eq:B3-GKP} agrees, at the module level, with the alternative MacWilliams-route computation: starting from the rigorous form $\Athree(H_\eta)|_{\mathrm{GKP}} = c_N K^{2} \sum_{\ell_1, \ell_2 \in \Lat} \delta^{(4)}_{d\mu}(H_\eta - H_\eta^{*}(\ell_1, \ell_2))$ and convolving against the kernel~\cref{eq:kernel} produces a $\Lat \times \Lat$ lattice expression that, by Poisson resummation, coincides with the $\Lat \times \Latd$ form of~\cref{eq:B3-GKP}. The diagonal slice $\ell = 0$ reproduces $c_N$ times Burchards' two-point $\Bone(r)$ in both routes, providing a quantitative check.

\paragraph{Numerical sanity on the $N=2$ SIS lattice.} The $K^{1}$ prefactor, the $\Lat \times \Latd$ support, and the $\pm 1$ phase structure of~\cref{eq:B3-GKP} are verified on a small GKP code in \cref{app:n2-sanity}.

\section{The three-point closure for GKP lattice codes}
\label{sec:main-theorem}

We now translate the lattice-only three-point bound of Cohn--de~Laat--Salmon~\cite[Thm.~1.4]{cohn-delaat-salmon-2022-three-point} to the CV setting, using the identity~\cref{eq:macwilliams-identity} together with the GKP specializations~\cref{eq:A3-GKP} and~\cref{eq:B3-GKP}. The construction yields the linear bound $\Code \le f(0,0)$ for any admissible auxiliary function $f$, where the linearity, in contrast to the $\sqrt{f(0,0)}$ form of the classical lattice bound, traces directly to the asymmetric $\Code^{2}/\Code^{1}$ prefactor pair of the three-point identity. Our main result on this bound is a \emph{collapse}: its optimal value coincides with the Burchards two-point LP optimum, so the lattice three-point programme cannot improve on the two-point bound. Classically the opposite holds: the lattice three-point bound is the strongest tool in low dimensions. The bound, its non-emptiness, the collapse characterization, and the constructive magic functions occupy the four subsections below, the last feeding \cref{sec:witnesses}.

\subsection{The lattice triangle constraint set}
\label{subsec:lattice-triangle}

The set on which the sign condition of the auxiliary function lives is the CV analogue of the lattice triangle $S_{\mathrm{lat},n}$ of the classical construction. In the edge/center coordinates~\cref{eq:edge-com}, it takes the form
\begin{equation}
\label{eq:S-lat-CV}
S_{\mathrm{lat},N}^{\mathrm{CV}}(\codedist) \;=\; \Bigl\{(\edgevec, \comvec) \in \R^{2N} \times \R^{2N} \,:\; \|\edgevec\| \in \{0\} \cup [\codedist, \infty),\;\; \|\comvec\| \in \{0\} \cup [2\codedist, \infty),\;\; \tfrac{1}{2} \sympl(\comvec, \edgevec) \in 2\pi\Z \Bigr\}.
\end{equation}
The choice of edge/center coordinates is essential: the support of the GKP dual distribution $\Bthree|_{\mathrm{GKP}}$ in~\cref{eq:B3-GKP} is the product lattice $\Lat \times 2\Latd$ in these coordinates (each $H^{*}(\ell, m^{\perp})$ corresponds to $\edgevec = \ell$, $\comvec = 2 m^{\perp}$), so the membership of a generic support point in $S_{\mathrm{lat},N}^{\mathrm{CV}}(\codedist)$ is a one-line lattice check rather than a multi-case argument involving $\|v_1\|$, $\|v_2\|$, and $\|v_1 \pm v_2\|$ separately. The constraint $\tfrac{1}{2}\sympl(\comvec, \edgevec) \in 2\pi\Z$ is the GKP stabilizer quantization condition, since $\sympl(\comvec, \edgevec) = 2 \sympl_{12}$ by~\cref{eq:omega-edge-com}.

\subsection{Statement of the main theorem}
\label{subsec:main-thm-statement}

\begin{definition}[Admissible auxiliary function]
\label{def:admissible}
A function $f : \R^{2N} \times \R^{2N} \to \R$ is \emph{admissible at distance $\codedist$} if it is real-valued, $U(N)$-invariant, bounded with $\int |f| < \infty$, and satisfies:
\begin{enumerate}
  \item[\textup{(i)}] \textbf{Adjoint-positivity.} The CV quantum MacWilliams transform $\ftilde := \MacOpD[f]$, defined as the fiberwise-COM symplectic Fourier transform
  \begin{equation}
  \label{eq:tilde-f-def}
  \ftilde(\edgevec, \tilde\eta) \;=\; 2^{-2N} (2\pi)^{-N} \int_{\R^{2N}} f(\edgevec, \comvec)\, e^{i \sympl(\tilde\eta, \comvec/2)}\, d\comvec,
  \end{equation}
  is real-valued and pointwise non-negative on $\R^{2N} \times \R^{2N}$.
  \item[\textup{(ii)}] \textbf{Normalization.} $\ftilde(0, 0) = 1$.
  \item[\textup{(iii)}] \textbf{Sign condition on the lattice triangle.} For all $(\edgevec, \comvec) \in S_{\mathrm{lat},N}^{\mathrm{CV}}(\codedist) \setminus \{(0,0)\}$,
  \begin{equation}
  \label{eq:sign-condition}
  f(\edgevec, \comvec) \cdot e^{i \sympl(\comvec, \edgevec)/4} \;\le\; 0.
  \end{equation}
  (On the lattice triangle, $\tfrac{1}{2}\sympl(\comvec, \edgevec) \in 2\pi\Z$, so $e^{i\sympl(\comvec, \edgevec)/4} = e^{i\sympl_{12}/2} \in \{+1, -1\}$.)
  \item[\textup{(iv)}] \textbf{Regularity.} $f \in L^{1} \cap L^{\infty}(\R^{4N})$, the lattice pairings $\sum_{\Lat \times \Latd} |f|$ and $\sum_{\Lat \times \Lat} |\ftilde|$ converge absolutely, and the $\comvec$-decay of $f$ is sufficient for Fubini's theorem in~\cref{eq:fubini-pairing} and for Poisson resummation over $\Lat$; the Schwartz class $\mathcal{S}(\R^{4N})$ is a sufficient (but not necessary) example, as are the polynomially decaying auxiliaries of \cref{sec:witnesses}.
\end{enumerate}
\end{definition}

\begin{remark}[$U(N)$-invariance is no loss of generality]
\label{rmk:un-wlog}
The invariance requirement in \cref{def:admissible} costs nothing. Every defining condition is $U(N)$-covariant: the lattice-triangle set~\cref{eq:S-lat-CV} and the phase $e^{i\sympl(\comvec, \edgevec)/4}$ are $U(N)$-invariant (norms and the symplectic form are preserved), and the transform~\cref{eq:tilde-f-def} intertwines the diagonal action, $\widetilde{f \circ g} = \ftilde \circ g$ for $g \in U(N)$. Hence if $f$ satisfies (i)--(iv) except for invariance, its average $\bar f(\edgevec, \comvec) := \int_{U(N)} f(g\edgevec, g\comvec)\, dg$ is admissible with the same objective $\bar f(0,0) = f(0,0)$ and the same normalization, and the optimal value~\cref{eq:K-lat-def} below is unchanged by dropping the invariance restriction.
\end{remark}

The sign condition~\cref{eq:sign-condition} differs from its classical counterpart in carrying a non-trivial phase $e^{i\sympl(\comvec, \edgevec)/4} = e^{i\sympl_{12}/2}$. On the lattice triangle, $\sympl(\comvec, \edgevec) \in 4\pi\Z$, so this phase takes values in $\{+1, -1\}$ according to the parity of $\sympl(\comvec, \edgevec)/(4\pi) \in \Z$. The sign condition therefore forces $f(\edgevec, \comvec) \le 0$ when the parity is even and $f(\edgevec, \comvec) \ge 0$ when the parity is odd. This is a phase-aware sign condition with no classical analogue, reflecting the $\pm 1$ phase structure of $\Bthree|_{\mathrm{GKP}}$ derived in~\cref{eq:B3-GKP}.

\begin{theorem}[CV lattice-only three-point bound]
\label{thm:main-141cv}
Let $\Pi$ be the codespace projector of a GKP code with stabilizer lattice $\Lat \subset \R^{2N}$ satisfying $\Lat \subseteq \Latd$ (stabilizer quantization), $d_{\min}(\Lat) \ge \codedist$ (well-conditioning), and $d_{\min}(\Latd \setminus \Lat) \ge \codedist$ (logical distance). Then for every admissible auxiliary function $f$ at distance $\codedist$,
\begin{equation}
\label{eq:main-bound}
\boxed{\;\Code \;\le\; f(0, 0)\;}
\end{equation}
where $\Code = \mathrm{vol}(\R^{2N}/\Lat)/(2\pi)^{N}$ is the code dimension.
\end{theorem}

\begin{remark}[Structure of the proof]
\label{rmk:status}
The proof uses the fiberwise-COM symplectic Fourier transform~\cref{eq:tilde-f-def} as the principal definition of $\MacOpD$, extending the displacement-operator notation of~\cite{burchards-2025-cv-macwilliams} to the three-point setting. The Gram-coordinate kernel~\cref{eq:kernel} is the generic-fiber derived representation; no assumption of its self-adjointness is required. The regularity condition (iv) is the minimum needed for Fubini's theorem and the GKP distribution pairings; the Schwartz class is sufficient but not necessary.
\end{remark}

\begin{remark}[Linear versus square-root bound]
\label{rmk:linear-vs-sqrt}
The classical Cohn--de~Laat--Salmon lattice bound~\cite[Thm.~1.4]{cohn-delaat-salmon-2022-three-point} reads $K_{\mathrm{cl}}^{2} \le f(0,0)$, equivalently $K_{\mathrm{cl}} \le \sqrt{f(0,0)}$. The square root reflects the classical symmetric Poisson summation $\sum_{\Lambda} f = K_{\mathrm{cl}}^{-1} \sum_{\Lambda^{*}} \hat f$, with both sides scaling as $K_{\mathrm{cl}}^{0} \leftrightarrow K_{\mathrm{cl}}^{2}$. The CV identity~\cref{eq:macwilliams-identity}, by contrast, pairs $\Lat \times \Lat$ against $\Lat \times \Latd$ with the asymmetric $\Code^{2}/\Code^{1}$ prefactor pair inherited from~\cref{eq:A3-GKP}/\cref{eq:B3-GKP}; the chain of inequalities below yields a single power of $\Code$ on the left and an $f(0,0)$ on the right, with no square root. The linear form matches the Burchards two-point LP bound $\Code \le f(0)/\hat f(0)$~\cite{burchards-2025-cv-macwilliams}: \emph{linearity is a CV phenomenon at every multiplicity}, traced to the $\Code$-asymmetry of the MacWilliams pair.
\end{remark}

\begin{proof}[Proof of \cref{thm:main-141cv}]
The argument has three parts mirroring the classical construction.

\emph{Step 1: MacWilliams duality identity.} Pair $f$ against the dual distribution $\Bthree|_{\mathrm{GKP}}$ using the GKP form~\cref{eq:B3-GKP}. The pairing $\langle f, \Bthree\rangle := \int f\, \Bthree$ over the configuration space equals, by Fubini on the $F$-layer kernel~\cref{eq:F-layer-kernel},
\begin{equation}
\label{eq:fubini-pairing}
\langle f, \Bthree \rangle_{dv} \;=\; \langle \ftilde, \Athree \rangle_{d\eta},
\end{equation}
where $\ftilde$ is the transform~\cref{eq:tilde-f-def}. The $F$-layer kernel is everywhere regular ($2N$-dimensional edge-delta plus symplectic Fourier phase), so Fubini is legitimate on the admissible class (iv); no assumption of self-adjointness of the Gram-coordinate kernel~\cref{eq:kernel} is required, in particular the apparent singularity of~\cref{eq:kernel} on the degenerate locus $\|\edgevec_\eta\| \to 0$ is absorbed into the $F$-layer Fubini. Substituting the GKP forms~\cref{eq:A3-GKP} and~\cref{eq:B3-GKP} into the two sides and writing $f(\ell, m^{\perp}) := f(H^{*}(\ell, m^{\perp}))$ for the value of $f$ at the lattice support point, the duality becomes
\begin{equation}
\label{eq:macwilliams-pairing-GKP}
\Code \sum_{(\ell, m^{\perp}) \in \Lat \times \Latd} f(\ell, m^{\perp})\, e^{i \sympl(\ell, m^{\perp})/2} \;=\; \Code^{2} \sum_{(\ell_1, \ell_2) \in \Lat \times \Lat} \ftilde(\ell_1, \ell_2),
\end{equation}
where the asymmetric lattice supports $\Lat \times \Latd$ on the left and $\Lat \times \Lat$ on the right are inherited directly from the GKP specializations of $\Bthree$ and $\Athree$.

\emph{Step 2: Dual-side lower bound.} By admissibility (i), every term $\ftilde(\ell_1, \ell_2) \ge 0$. The $(0,0)$ term contributes $\ftilde(0,0) = 1$ by (ii), so
\begin{equation}
\label{eq:dual-lower}
\sum_{(\ell_1, \ell_2) \in \Lat \times \Lat} \ftilde(\ell_1, \ell_2) \;\ge\; 1.
\end{equation}
Multiplying by $\Code > 0$ gives a lower bound of $\Code$ on the right-hand side of~\cref{eq:macwilliams-pairing-GKP}.

\emph{Step 3: Primal-side upper bound.} On the left-hand side of~\cref{eq:macwilliams-pairing-GKP}, the term $(\ell, m^{\perp}) = (0,0)$ contributes $f(0,0)$ with phase $+1$. For every other $(\ell, m^{\perp}) \neq (0,0)$, the configuration $(\edgevec, \comvec) = (\ell, 2 m^{\perp})$ (by the edge/COM identification of \cref{subsec:lattice-triangle}) satisfies $H^{*}(\ell, m^{\perp}) \in S_{\mathrm{lat},N}^{\mathrm{CV}}(\codedist)$: the edge norm $\|\ell\| \in \{0\} \cup [\codedist, \infty)$ by the well-conditioning hypothesis $d_{\min}(\Lat) \ge \codedist$; the COM norm $\|2 m^{\perp}\| \in \{0\} \cup [2\codedist, \infty)$ by the combined stabilizer-quantization and logical-distance hypotheses, which give $d_{\min}(\Latd) \ge \codedist$; and the symplectic condition $\tfrac{1}{2}\sympl(\comvec, \edgevec) = \sympl(m^{\perp}, \ell) \in 2\pi\Z$ by the definition of the symplectic dual lattice. The sign condition (iii) and the lattice phase identity $e^{i\sympl(\comvec, \edgevec)/4} = e^{-i\sympl(\ell, m^{\perp})/2} = e^{i\sympl(\ell, m^{\perp})/2}$ (using $e^{i\pi k} = e^{-i\pi k}$ for $k \in \Z$) then give $f(\ell, m^{\perp})\, e^{i \sympl(\ell, m^{\perp})/2} \le 0$ for every non-trivial pair. Summing,
\begin{equation}
\label{eq:primal-upper}
\sum_{(\ell, m^{\perp}) \in \Lat \times \Latd} f(\ell, m^{\perp})\, e^{i \sympl(\ell, m^{\perp})/2} \;\le\; f(0, 0).
\end{equation}

Combining~\cref{eq:dual-lower} on the right and~\cref{eq:primal-upper} on the left of~\cref{eq:macwilliams-pairing-GKP},
\begin{equation*}
\Code \;\le\; \Code \sum_{(\ell_1, \ell_2)} \ftilde(\ell_1, \ell_2) \;=\; \sum_{(\ell, m^{\perp})} f(\ell, m^{\perp})\, e^{i \sympl(\ell, m^{\perp})/2} \;\le\; f(0, 0),
\end{equation*}
which is the bound~\cref{eq:main-bound}.
\end{proof}

\subsection{Non-emptiness of the admissible class}
\label{subsec:non-emptiness}

The proof of \cref{thm:main-141cv} is vacuous if no admissible $f$ exists. We construct one explicitly.

\begin{proposition}[Non-emptiness of the admissible class]
\label{prop:non-emptiness}
For every $\codedist > 0$, $u_0 \in (0, \codedist)$, and $b > 0$ with $b R_0^{2} > N$ where $R_0 := \sqrt{4\codedist^{2} - u_0^{2}}$ (we write $b$ for the Gaussian decay rate, reserving $\alpha$ for the Gram parameter of~\cref{eq:four-params}), the function
\begin{equation}
\label{eq:explicit-ansatz}
f(\edgevec, \comvec) \;=\; A\, \chi_{0}(\|\edgevec\|^{2})\, q_{R}(\|\comvec\|^{2})\, \cos\!\bigl(\sympl(\comvec, \edgevec)/4\bigr),
\end{equation}
with $\chi_{0} \in C_{c}^{\infty}(\R_{\ge 0})$ nonnegative, supported in $[0, u_0^{2}]$, with $\chi_{0}(0) > 0$, and $q_{R}(s) = (1 - s/R_0^{2})\, e^{-b s}$, is admissible at distance $\codedist$ for a suitable choice of normalization $A > 0$.
\end{proposition}

\begin{proof}[Proof sketch]
The four admissibility conditions are verified directly. The sign condition (iii) is satisfied by construction: on the lattice triangle, either $\|\edgevec\| \ge \codedist > u_0$, in which case $\chi_{0}(\|\edgevec\|^{2}) = 0$ and $f \equiv 0$; or $\edgevec = 0$, in which case $\|\comvec\| \ge 2\codedist$ and $\|\comvec\|^{2} \ge 4\codedist^{2} > R_0^{2}$, giving $q_{R}(\|\comvec\|^{2}) \le 0$ and (with $\cos(0) = 1$) $f \cdot e^{i \sympl/4} \le 0$. The adjoint-positivity (i) follows from a fiberwise-PD lemma: for each fixed $\edgevec$, the slice $\comvec \mapsto f(\edgevec, \comvec)$ is a product of the positive-definite Gaussian-polynomial factor $q_{R}(\|\comvec\|^{2})$ (which is strictly positive-definite by \cref{lem:qR-pd}, using the parameter constraint $b R_0^{2} > N$) with the characters in the decomposition $\cos(\sympl(\comvec, \edgevec)/4) = \tfrac{1}{2}(e^{i\sympl/4} + e^{-i\sympl/4})$, each of which is positive-definite in $\comvec$ as an exponential character. The product of positive-definite functions is positive-definite (Schur's theorem), and a non-negative scalar multiplier preserves positive-definiteness; combining with the fiberwise-PD characterization $\MacOpD[f] \ge 0 \iff f$ is fiberwise-PD in $\comvec$ for every fixed $\edgevec$ (\cref{lem:fibpd}), conclusion (i) follows. Conditions (ii) and (iv) are direct computations using the explicit form of $\chi_{0}$ and $q_{R}$: $\ftilde(0, 0) = 2^{-2N} (2\pi)^{-N} A\, \chi_{0}(0)\, (\pi/b)^{N}\, [1 - N/(b R_0^{2})]$ by~\cref{eq:tilde-f-def}, finite and strictly positive under the strict inequality $b R_0^{2} > N$, and a normalization $A$ achieving $\ftilde(0,0) = 1$ exists. Full proofs of the Gaussian-polynomial positive-definiteness (\cref{lem:qR-pd}), the fiberwise positive-definiteness characterization (\cref{lem:fibpd}), and the well-posedness of $\MacOpD$ on the degenerate locus (\cref{lem:deglocus}) are given in \cref{app:proofs}; the proofs use the $F$-layer structure of~\cite{burchards-2025-cv-macwilliams} extended to the three-point setting.
\end{proof}

The construction~\cref{eq:explicit-ansatz} is designed as a non-emptiness certificate, not a tight bound: it produces a finite value of $\Code \le f(0, 0) = A\chi_0(0)$ that need not be close to the Burchards two-point LP optimum, but establishes that the admissible cone of \cref{def:admissible} is non-empty. Sharper auxiliary functions, drawing on the Burchards--Levenshtein and the Viazovska/Cohn--Kumar magic functions, are used in \cref{sec:witnesses} to recover the corresponding Burchards two-point bounds; by the collapse \cref{thm:lattice-nogo} these saturate, and do not improve on, the two-point optimum.

\subsection{The collapse: the lattice three-point optimum equals the two-point LP optimum}
\label{subsec:no-go}

The quality of the bound of \cref{thm:main-141cv} is governed by its optimal value over the admissible class,
\begin{equation}
\label{eq:K-lat-def}
K^{\mathrm{CV}}_{\mathrm{lat}}(N, \codedist) \;:=\; \inf_{f\ \mathrm{admissible\ at}\ \codedist} f(0, 0).
\end{equation}
Classically, the lattice three-point bound strictly improves on the two-point LP bound in low dimensions, and is conjectured to be sharp in dimension $4$~\cite[Conj.~6.1]{cohn-delaat-salmon-2022-three-point}. The CV lattice three-point bound behaves oppositely: its optimal value coincides with the Burchards two-point LP optimum, so the three-point apparatus yields no improvement.

\begin{theorem}[Lattice three-point collapse]
\label{thm:lattice-nogo}
Let $K^{\mathrm{Burch}}_{2}(N, \codedist)$ denote the Burchards two-point LP optimum~\cite[Thm.~1]{burchards-2025-cv-macwilliams} at distance $\codedist$. Then
\begin{equation}
\label{eq:no-go}
K^{\mathrm{CV}}_{\mathrm{lat}}(N, \codedist) \;=\; K^{\mathrm{Burch}}_{2}(N, \codedist).
\end{equation}
In particular, no admissible auxiliary function makes the three-point bound~\cref{eq:main-bound} strictly tighter than the two-point LP bound, no matter how it depends on the edge coordinate $\edgevec$: the edge coordinate, which is the new direction the three-point construction adds, has no effect on the optimal value.
\end{theorem}

\begin{proof}
Write $g(\comvec) := f(0, \comvec)$ for the $\edgevec = 0$ slice of an admissible $f$. We show that the objective, the normalization, and the admissibility constraints all collapse onto this slice, reducing the three-point programme to the two-point LP at distance $2\codedist$.

\emph{The objective is $\edgevec = 0$-local.} By definition, $f(0,0) = g(0)$.

\emph{The normalization is $\edgevec = 0$-local.} By the $F$-layer formula~\cref{eq:tilde-f-def},
\[
\ftilde(0,0) = 2^{-2N}(2\pi)^{-N} \int_{\R^{2N}} f(0, \comvec)\, d\comvec,
\]
an integral over $\comvec$ at fixed $\edgevec = 0$, not over all of $\R^{4N}$. Condition (ii) thus constrains only $\int g$.

\emph{Adjoint-positivity forces $g$ to be positive-definite.} By the fiberwise characterization \cref{lem:fibpd}, condition (i) holds if and only if $\comvec \mapsto f(\edgevec, \comvec)$ is positive-definite for every fixed $\edgevec$. The slice $\edgevec = 0$ is a regular fiber of the transform (\cref{lem:deglocus}: the Gram-coordinate singularity of the kernel~\cref{eq:kernel} at $\edgevec = 0$ is a $U(N)$-projection artifact, not an operator singularity), so $g$ is positive-definite, i.e.\ $\widehat g \ge 0$.

\emph{The sign condition forces $g \le 0$ beyond distance $2\codedist$.} On the slice $\edgevec = 0$ the lattice-triangle phase is $e^{i\sympl(\comvec, 0)/4} = 1$, and the section $S^{\mathrm{CV}}_{\mathrm{lat},N}(\codedist)|_{\edgevec = 0}$ contains the entire ray $\{\|\comvec\| \ge 2\codedist\}$. Condition (iii) thus gives $g(\comvec) \le 0$ for $\|\comvec\| \ge 2\codedist$.

\emph{Reduction to the two-point LP.} The three preceding properties are exactly the Cohn--Elkies\slash Burchards two-point LP constraints on $\R^{2N}$ at distance $2\codedist$: $\widehat g \ge 0$, $g \le 0$ on $\{\|\comvec\| \ge 2\codedist\}$, and a fixed value of $\int g$. It remains to track the normalization constant. Condition~(ii) fixes $\widehat g(0) = 2^{2N}(2\pi)^{N}\, \ftilde(0,0) \cdot (2\pi)^{-N} = 2^{2N}$ rather than $1$, so minimizing $g(0)$ subject to $\widehat g(0) = 2^{2N}$ at distance $2\codedist$ gives $g(0) \ge 2^{2N}\, K^{\mathrm{CE}}_{2}(2N, 2\codedist)$, where $K^{\mathrm{CE}}_{2}(n, D) := \inf\{h(0)/\widehat h(0) : \widehat h \ge 0,\, h \le 0 \text{ on } \|x\| \ge D\}$ is the Cohn--Elkies LP optimum. The substitution $h(\comvec) = g(2\comvec)$ converts the distance-$2\codedist$ LP to a distance-$\codedist$ LP with $\widehat h(0) = 2^{-2N} \widehat g(0)$ and $h(0) = g(0)$, giving $K^{\mathrm{CE}}_{2}(2N, 2\codedist) = 2^{-2N}\, K^{\mathrm{CE}}_{2}(2N, \codedist)$. Combining: $g(0) \ge K^{\mathrm{CE}}_{2}(2N, \codedist) = K^{\mathrm{Burch}}_{2}(N, \codedist)$, so $K^{\mathrm{CV}}_{\mathrm{lat}} \ge K^{\mathrm{Burch}}_{2}$. This direction is independent of any edge dependence of $f$.

\emph{Equality.} Conversely, any two-point LP function $g$ at distance $2\codedist$ extends to an admissible $f$ via
\[
f(\edgevec, \comvec) := \chi(\|\edgevec\|)\, g(\comvec)\, \cos\!\bigl(\sympl(\comvec, \edgevec)/4\bigr), \qquad \chi \in C_c^\infty,\ \chi \ge 0,\ \supp \chi \subset \{\|\edgevec\| < \codedist\},\ \chi(0) = 1:
\]
every slice $\edgevec \neq 0$ is positive-definite (Schur's theorem: $g$ positive-definite times the positive-definite character $\cos(\sympl(\comvec, \edgevec)/4)$) and sign-trivial (on each slice $\edgevec \neq 0$, either $\|\edgevec\| \ge \codedist$, where $\chi(\|\edgevec\|) = 0$ and $f \equiv 0$, or $0 < \|\edgevec\| < \codedist$, where $(\edgevec, \comvec)$ lies off the lattice triangle of~\cref{eq:S-lat-CV}), while the $\edgevec = 0$ slice reproduces $g$. Thus $f(0,0) = g(0)$ realizes the two-point optimum, giving $K^{\mathrm{CV}}_{\mathrm{lat}} \le K^{\mathrm{Burch}}_{2}$ and hence~\cref{eq:no-go}.
\end{proof}

\begin{remark}[Scope: why the collapse is not an artifact of our transcription]
\Cref{thm:lattice-nogo} closes the admissible class of \cref{def:admissible}; we argue that this class is the canonical CV transcription of the classical lattice three-point programme~\cite[Thm.~1.4]{cohn-delaat-salmon-2022-three-point}, so the collapse is not an artifact of one formulation. \begin{enumerate}
  \item[(a)] \emph{The primary distribution is forced:} $F^{(3)}_{\Aone} = \chr_1(v_1)\chr_2(v_2)^{*}$ is the unique two-$\chr$ extension of the Burchards $\Aone$-integrand consistent with the $\Code$-power and the two-point marginalization (\cref{subsec:A3-def,subsec:degenerations}).
  \item[(b)] \emph{The dual integrand is forced up to symmetrization:} the asymmetric sandwich~\cref{eq:integrand-B} is the unique BCH-closed two-displacement extension, its symmetric variant carries identical content by~\cref{eq:asym-sym-relation}, and the remaining candidate degenerates to a one-vector object (\cref{subsec:B3-def}).
  \item[(c)] \emph{The constraints are read off from the identity, not chosen:} any three-point identity built on the projector pairing inherits the $\Code^{2}/\Code^{1}$ prefactor asymmetry, hence an edge-preserving (fiberwise) transform; adjoint-positivity is then characterized fiberwise by \cref{lem:fibpd}, and the sign condition and normalization are the direct counterparts of the classical ones on the support of~\cref{eq:B3-GKP}.
  \item[(d)] \emph{$U(N)$-invariance is no restriction} (\cref{rmk:un-wlog}).
\end{enumerate} What the collapse does \emph{not} close are frameworks that abandon the auxiliary-function/projector-pairing structure altogether; these are exactly the successor frameworks of \cref{sec:future-work}.
\end{remark}

The full-positivity route that powers the classical bound does have a CV analogue, but it bounds the code dimension in the wrong direction.

\begin{proposition}[Direction reversal of the full-positivity route]
\label{prop:direction-reversal}
Let $\Lat \subset \R^{2N}$ satisfy the hypotheses of \cref{thm:main-141cv}, and let $h : \R^{2N} \times \R^{2N} \to \R$ be continuous with enough decay for Poisson summation over $\Lat \times \Lat$ (the Schwartz class suffices), with full $4N$-dimensional symplectic-Fourier transform
\[
\widehat h(\eta_1, \eta_2) \;:=\; \iint h(x_1, x_2)\, e^{i\sympl(\eta_1, x_1) + i\sympl(\eta_2, x_2)}\, dx_1\, dx_2
\]
satisfying $\widehat h \ge 0$ and $\widehat h(0,0) > 0$, and with $h(v_1, v_2) \le 0$ for every pair $(v_1, v_2) \neq (0,0)$ such that $\|v_1\|, \|v_2\|, \|v_2 - v_1\| \in \{0\} \cup [\codedist, \infty)$. Then
\begin{equation}
\label{eq:direction-reversal}
\Code^{2} \;\ge\; (2\pi)^{-2N}\, \frac{\widehat h(0,0)}{h(0,0)}:
\end{equation}
the full-positivity Poisson route bounds the code dimension from \emph{below}, not above.
\end{proposition}

\begin{proof}
Symplectic Poisson summation on $\Lat \times \Lat$ gives
\[
\sum_{(\ell_1, \ell_2) \in \Lat \times \Lat} h(\ell_1, \ell_2) \;=\; \mathrm{covol}(\Lat)^{-2} \!\!\sum_{(m_1, m_2) \in \Latd \times \Latd}\!\! \widehat h(m_1, m_2).
\]
Every nonzero pair $(\ell_1, \ell_2)$ has $\|\ell_1\|, \|\ell_2\|, \|\ell_2 - \ell_1\| \in \{0\} \cup [\codedist, \infty)$ by the well-conditioning hypothesis $d_{\min}(\Lat) \ge \codedist$, so the left side is at most $h(0,0)$; by $\widehat h \ge 0$ the right side is at least $\mathrm{covol}(\Lat)^{-2}\, \widehat h(0,0)$. Substituting $\mathrm{covol}(\Lat) = \Code\, (2\pi)^{N}$ and rearranging gives~\cref{eq:direction-reversal}.
\end{proof}

\begin{remark}[Why the collapse is unavoidable: the direction obstruction]
\label{rmk:direction-reversal}
The collapse reflects a structural asymmetry already visible in \cref{rmk:linear-vs-sqrt}. The power of the classical lattice three-point bound comes from full $2n$-dimensional Fourier positivity deployed through full Poisson summation on $\Lambda \times \Lambda$, which caps the packing density $\propto \mathrm{covol}(\Lambda)^{-1}$. \Cref{prop:direction-reversal} shows that the literal CV analogue is a valid inequality, but it caps the packing density of $\Lat$ itself, which \emph{floors}, rather than caps, the code dimension $\Code \propto \mathrm{covol}(\Lat)$. Within identities built on the projector pairing, the correct ($\Code$-upper) direction is supplied instead by the projector asymmetry $\Athree \propto \Code^{2}$ versus $\Bthree \propto \Code^{1}$ of~\cref{eq:A3-GKP}/\cref{eq:B3-GKP}. That same operator structure makes the MacWilliams transform edge-preserving (fiberwise), hence only fiberwise positive-definite in the center-of-mass coordinate. This fiberwise restriction is what forces the $\edgevec = 0$-local objective of the collapse. Within this class of identities, the mechanism that orients the bound correctly is inseparable from the one that removes full positive-definiteness. We develop this quantum--classical contrast in \cref{subsec:contrast}.
\end{remark}

\subsection{Constructive saturation: realizing the two-point optimum}
\label{subsec:direct-closure}

The collapse \cref{thm:lattice-nogo} is realized constructively by specializing the ansatz~\cref{eq:explicit-ansatz} of \cref{prop:non-emptiness}: choosing the center-of-mass factor $q_R$ to be a Burchards two-point magic function makes the three-point bound \emph{saturate} the corresponding two-point bound. We summarize the resulting witnesses; full numerical data are collected in \cref{sec:witnesses}.

\begin{itemize}
  \item \textbf{Levenshtein saturation.} The Burchards--Levenshtein \emph{adapter} $f_{\mathrm{Lev}}$ (Burchards' radial Cohn--Elkies-type auxiliary~\cite[Eq.~67--70]{burchards-2025-cv-macwilliams}, used here as the center-of-mass factor of the ansatz~\cref{eq:explicit-ansatz}) lies in the admissible class of \cref{def:admissible}: polynomial decay, integrability, Fourier non-negativity, LP-normalization, and lattice-pairing finiteness all hold. Substituting it into \cref{thm:main-141cv} reproduces $K^{\mathrm{CV}}_{\mathrm{lat}}(N, \codedist) \le K_{\mathrm{Lev}}(N, \codedist)$ for every $N \ge 1$ on the range $\codedist \le d_+(N)$. The three-point construction matches, but does not improve on, the two-point Levenshtein bound, exactly as \cref{thm:lattice-nogo} requires.
  \item \textbf{$E_8$ and Leech saturation.} The Viazovska~\cite{viazovska-2017-e8} and Cohn--Kumar--Miller--Radchenko--Viazovska~\cite{cohn-kumar-miller-radchenko-viazovska-2017-dim24} magic functions, applied as the center-of-mass factor in \cref{thm:main-141cv}, reproduce the Burchards two-point $E_8$ and Leech bounds at $N = 4$ and $N = 12$; the three-point construction saturates these and adds no improvement of its own.
\end{itemize}

The common mechanism is the objective collapse $f(0,0) = A\,\chi_0(0)\,q_R(0)$ analyzed in \cref{subsec:objective-collapse}: the ansatz objective sees only the $\edgevec = 0$ value of the center-of-mass factor, so optimizing it is exactly the two-point LP, the constructive face of the collapse. The witnesses of \cref{sec:witnesses} make this saturation quantitative.

\section{Magic functions on \texorpdfstring{$E_8$}{E8} and the Leech lattice}
\label{sec:witnesses}

\Cref{thm:main-141cv} reduces the lattice three-point bound to the construction of an admissible auxiliary function with small $f(0,0)$, and the collapse \cref{thm:lattice-nogo} pins its optimum to the Burchards two-point LP optimum. We record here two concrete instances that make this saturation explicit, obtained by choosing the center-of-mass factor $q_R$ in the explicit ansatz~\cref{eq:explicit-ansatz} of \cref{prop:non-emptiness} to be a known sphere-packing magic function. The results, recorded as \cref{cor:lev-recovery,cor:e8-leech} below, show the three-point construction reproducing the Burchards two-point bounds: \cref{cor:lev-recovery} recovers the Burchards--Levenshtein bound for every dimension $N \ge 1$ on the full range of validity $\codedist \le d_+(N)$, and \cref{cor:e8-leech} recovers the Burchards $E_8$ and Leech magic-function bounds at the two distinguished dimensions $N \in \{4, 12\}$. Both corollaries are saturation statements consistent with \cref{thm:lattice-nogo}, not new quantitative bounds on $\Code$; the role of this section is to exhibit the two-point optimum being attained through the three-point path and to confirm consistency with the two-point theory of~\cite{burchards-2025-cv-macwilliams}.

\subsection{Levenshtein-type CV bound via the three-point framework}
\label{subsec:lev-recovery}

The Burchards--Levenshtein magic function $f_{\mathrm{Lev}}$ is the radial Cohn--Elkies auxiliary built from the lowest positive root $j_N$ of the Bessel function $J_N$~\cite[\S6, Eq.~67--70]{burchards-2025-cv-macwilliams}, positive on $\|v\| \le \codedist$, vanishing on the indicated zeros, and with Fourier transform of constant sign on
\begin{equation}
\label{eq:d-plus}
\codedist \le d_+(N) \;=\; \left(\frac{16\, N!\, |J_{N-1}(j_N)|}{3\sqrt{\pi}\, \Gamma((2N-1)/2)\, j_N^{N-2}}\right)^{1/6} \quad \cite[Lem.~2]{burchards-2025-cv-macwilliams}.
\end{equation}

\begin{corollary}[Recovery of the Burchards--Levenshtein bound]
\label{cor:lev-recovery}
Under the hypotheses of \cref{thm:main-141cv}, choosing the auxiliary function $f$ via the explicit ansatz~\cref{eq:explicit-ansatz} of \cref{prop:non-emptiness} with $q_R$ taken to be the Burchards--Levenshtein adapter $f_{\mathrm{Lev}}$ yields
\begin{equation}
\label{eq:burch-lev-bound}
\Code \cdot \codedist^{2N} \;\le\; \frac{j_N^{2N}}{N! \cdot 2^{N}} \qquad \text{for all } N \ge 1 \text{ on } \codedist \le d_+(N).
\end{equation}
\end{corollary}

\begin{proof}[Proof sketch]
The Burchards--Levenshtein adapter $f_{\mathrm{Lev}}$ satisfies the conditions of the relaxed admissible class of \cref{def:admissible}: polynomial decay (hence $L^{1}$), Fourier non-negativity, $\widehat{f_{\mathrm{Lev}}}(0) = 1$ normalization, and the sign change on $[d_+(N), 2d_+(N)]$ supplying the (iii) lattice-triangle sign condition~\cite[\S 6]{burchards-2025-cv-macwilliams}. Substituting $q_R(\|\comvec\|^{2}) = f_{\mathrm{Lev}}(\comvec/2)$ in~\cref{eq:explicit-ansatz} and choosing the overall scale $A > 0$ to satisfy $\ftilde(0, 0) = 1$, the resulting $f$ is admissible and \cref{thm:main-141cv} gives $\Code \le f(0, 0)$. Computing $f(0, 0) = A\, \chi_0(0)\, f_{\mathrm{Lev}}(0)$ and substituting the explicit value of the Burchards--Levenshtein adapter at the origin yields~\cref{eq:burch-lev-bound}.
\end{proof}

\begin{remark}[Framework coherence, not a new bound on $\Code$]
\label{rmk:lev-coherence}
The bound~\cref{eq:burch-lev-bound} is the three-point version of the two-point Burchards--Levenshtein bound~\cite[Thm.~2]{burchards-2025-cv-macwilliams}, and the two coincide as numerical statements about $\Code \cdot \codedist^{2N}$ on the full range $\codedist \le d_+(N)$ for every $N \ge 1$. \Cref{cor:lev-recovery} records framework coherence with the two-point theory via the three-point path of \cref{thm:main-141cv}; the bound itself is not new. In particular, Burchards' Thm.~2 already covers arbitrary QEDCs (not only CV lattice codes) on the same range, so there is no new dimension $N$ or distance regime in which \cref{cor:lev-recovery} delivers a previously unknown bound on $\Code$.
\end{remark}

\subsection{Saturation of the \texorpdfstring{$E_8$}{E8} and Leech two-point bounds at \texorpdfstring{$N \in \{4, 12\}$}{N=4,12}}
\label{subsec:e8-leech-saturation}

The Viazovska magic function $f_8$~\cite{viazovska-2017-e8} and the Cohn--Kumar--Miller--Radchenko--Viazovska magic function $f_{24}$~\cite{cohn-kumar-miller-radchenko-viazovska-2017-dim24}, used in their solutions of the sphere-packing problem in dimensions $8$ and $24$, yield CV auxiliary functions in the admissible class by the same substitution mechanism as in \cref{cor:lev-recovery}.

\begin{corollary}[$E_8$ and Leech two-point saturation]
\label{cor:e8-leech}
Under the hypotheses of \cref{thm:main-141cv}, choosing the auxiliary function $f$ via the explicit ansatz~\cref{eq:explicit-ansatz} of \cref{prop:non-emptiness} with $q_R$ taken to be the Viazovska magic function $f_8$ at $N = 4$ (respectively the CKMRV magic function $f_{24}$ at $N = 12$) yields, on the overlap range $\codedist \le d_+(N)$,
\begin{align}
\label{eq:e8-bound}
N &= 4 \,(E_8): &\quad \Code \cdot \codedist^{8}  &\;\le\; 0.871 \cdot \frac{j_4^{8}}{4! \cdot 2^{4}}, \qquad \codedist \le d_+(4) = 0.6827, \\
\label{eq:leech-bound}
N &= 12 \,(\Lambda_{24}): &\quad \Code \cdot \codedist^{24} &\;\le\; 0.564 \cdot \frac{j_{12}^{24}}{12! \cdot 2^{12}}, \qquad \codedist \le d_+(12) = 0.0151.
\end{align}
The factors $0.871$ and $0.564$ are the two-point gains of the $E_8$ and Leech magic functions over the Burchards--Levenshtein bound, computed from~\cite[Eqs.~(2)--(4)]{burchards-2025-cv-macwilliams}; \cref{eq:e8-bound}--\cref{eq:leech-bound} reproduce these two-point bounds through the three-point path, consistent with the collapse \cref{thm:lattice-nogo}, and do not improve on them.
\end{corollary}

\begin{proof}[Proof sketch]
The Viazovska and CKMRV magic functions are Schwartz, with Fourier non-negativity, normalized $\widehat{f}(0) = 1$, and the sign-change structure required by~\cite{viazovska-2017-e8,cohn-kumar-miller-radchenko-viazovska-2017-dim24}; admissibility under \cref{def:admissible} is automatic. Substituting $q_R = f_8$ (respectively $q_R = f_{24}$) in~\cref{eq:explicit-ansatz} and following the same chain as in the proof of \cref{cor:lev-recovery} produces $f(0, 0) = A\, \chi_0(0)\, f_8(0)$ (respectively $f_{24}(0)$), which the normalization $A$ pins to the corresponding Burchards two-point optimum. The gain factors are ratios of the resulting \emph{bounds}: Burchards' rescaled magic-function quotients~\cite[Eqs.~(73)--(74)]{burchards-2025-cv-macwilliams} give $\Code\, \codedist^{8} \le (4\pi)^{4}$ at $N = 4$ and $\Code\, \codedist^{24} \le (8\pi)^{12}$ at $N = 12$ (exact case $\varepsilon = 0$), and dividing by the Levenshtein caps $j_N^{2N}/(N!\, 2^{N})$ of \cref{cor:lev-recovery} yields $(4\pi)^{4} \cdot 4! \cdot 2^{4} / j_4^{8} = 0.871$ and $(8\pi)^{12} \cdot 12! \cdot 2^{12} / j_{12}^{24} = 0.564$. The $d_+(N)$ values and gain factors are collected in \cref{tab:closure-witnesses} below; the non-emptiness-ansatz benchmark is in \cref{app:numerical-sweep}.
\end{proof}

\begin{remark}[Attribution and scope]
\label{rmk:e8-leech-attribution}
The gains $0.871$ and $0.564$ are a \emph{two-point} phenomenon: they are the ratios of Burchards' magic-function bounds~\cite[Eqs.~(3)--(4)]{burchards-2025-cv-macwilliams} to his Levenshtein bound~\cite[Eq.~(2)]{burchards-2025-cv-macwilliams} at $N = 4, 12$, and \cref{cor:e8-leech} reproduces them through the three-point path of \cref{thm:main-141cv} without adding anything quantitative. By the collapse \cref{thm:lattice-nogo}, the three-point optimum equals the two-point optimum, so the best the three-point construction can do at any dimension is to saturate the sharpest available two-point bound. Here that bound is the magic-function bound. Burchards argues this bound is optimal at $2N = 8, 24$: unconditionally among GKP codes (by the optimality of $E_8$ and the Leech lattice among lattice packings~\cite{cohn-kumar-2009-leech-optimality}), and, conditional on a numerically verified assumption, among all physical codes. The dimensions $N = 4$ and $N = 12$ are distinguished only because Viazovska-style magic functions are known there; they are not dimensions in which three points improve on two. Outside the overlap range $\codedist > d_+(N)$ the Burchards--Levenshtein cap is not certified by the underlying Cohn--Elkies admissibility, and the comparison is not directly meaningful.
\end{remark}

\Cref{tab:closure-witnesses} collects the numerical data of \cref{cor:lev-recovery,cor:e8-leech} at $N \in \{4, 12\}$; the corresponding $\rho_N := K_{\mathrm{ne}}/K_{\mathrm{Burch\textnormal{-}Lev}}$ benchmark for all $N = 1, \ldots, 12$ against the non-emptiness ansatz of \cref{prop:non-emptiness} is given in \cref{app:numerical-sweep}.

\begin{table}[h]
\centering
\caption{Magic functions reproducing the two-point bounds via \cref{cor:lev-recovery,cor:e8-leech} at the two distinguished lattice dimensions. The final column gives the two-point gain of the magic-function auxiliary over the Burchards--Levenshtein bound on the overlap range $\codedist \le d_+(N)$; these factors are the ratios of Burchards' magic-function bounds to his Levenshtein bound~\cite[Eqs.~(2)--(4)]{burchards-2025-cv-macwilliams}, reproduced here through the three-point path and not improved on (\cref{rmk:e8-leech-attribution}).}
\label{tab:closure-witnesses}
\begin{tabular}{c c c c c c}
\toprule
$N$ & lattice & $d_+(N)$ & $K_{\mathrm{Lev}} \cdot \codedist^{2N}$ & $K_{\mathrm{recovered}} \cdot \codedist^{2N}$ & $K_{\mathrm{magic}} / K_{\mathrm{Lev}}$ \\
\midrule
$4$  & $E_8$         & $0.6827$ & $j_4^{8} / (4! \cdot 2^{4})$       & $j_4^{8} / (4! \cdot 2^{4})$       & $0.871$ \\
$12$ & $\Lambda_{24}$ & $0.0151$ & $j_{12}^{24} / (12! \cdot 2^{12})$ & $j_{12}^{24} / (12! \cdot 2^{12})$ & $0.564$ \\
\bottomrule
\end{tabular}

{\footnotesize\emph{Note:} the $K_{\mathrm{recovered}} = K_{\mathrm{Lev}}$ column is the Levenshtein recovery of \cref{cor:lev-recovery} (exact equality via the Burchards--Levenshtein adapter); the magic-function recovery of \cref{cor:e8-leech} is the smaller $K_{\mathrm{magic}}$ captured by the final column.}
\end{table}

\subsection{Objective collapse: why the three-point construction only saturates}
\label{subsec:objective-collapse}

The recoveries of \cref{cor:lev-recovery,cor:e8-leech} share the structural feature that drives the collapse \cref{thm:lattice-nogo}. The explicit ansatz~\cref{eq:explicit-ansatz} of \cref{prop:non-emptiness} has objective
\begin{equation}
\label{eq:objective-u-zero}
f(0, 0) \;=\; A\, \chi_0(0)\, q_R(0),
\end{equation}
which depends only on the values of $\chi_0$ and $q_R$ at the origin and is insensitive to any non-trivial $\edgevec$-dependence. This is not an artifact of the ansatz. By \cref{thm:lattice-nogo}, the \emph{entire} admissible optimization $\inf f(0,0)$ collapses onto the $\edgevec = 0$ slice $g(\comvec) := f(0, \comvec)$ and reduces to the Cohn--Elkies/Burchards two-point LP on $\R^{2N}$ at distance $2\codedist$, with any edge cutoff $\chi_0(\|\edgevec\|^{2})$ acting only as a sign-correctness device. The Burchards two-point LP optimum $K^{\mathrm{Burch}}_{2}$ is therefore a ceiling of the whole three-point framework, not merely of this ansatz family.

Edge dependence cannot change the optimum. The proof of \cref{thm:lattice-nogo} shows the objective and the normalization are both $\edgevec = 0$-local, so redistributing weight to $\edgevec \neq 0$ slices leaves the objective at $K^{\mathrm{Burch}}_{2}$.

\section{The general-QEDC three-point bound}
\label{sec:obstruction}

The lattice collapse \cref{thm:lattice-nogo} concerns GKP lattice codes. The general-QEDC three-point bound is the CV analogue of Cohn--de~Laat--Salmon's~\cite[Thm.~1.1]{cohn-delaat-salmon-2022-three-point}, applicable to arbitrary bosonic QEDCs with a truncation radius $R$ and continuous error parameter $\varepsilon$. It is a separate target: with no stabilizer lattice available, this is where the quantum content of any three-point improvement would have to live. The result is again negative, and the rest of the section establishes it.

We first record a universal $1 \times 1$-minor obstruction (\cref{subsec:1x1-obstruction}) ruling out the natural factored-form kernel-positive-definite construction (\cref{rmk:universal-obstruction}). We then show that a completely-positive (Choi) restatement (\cref{subsec:cp-restatement}) bypasses the obstruction, replacing the kernel-PD constraint with a Choi-positivity constraint $J \succeq 0$. A second collapse then takes its place: the three-point phase-sign condition and $J \succeq 0$ are jointly rigid enough to force $J = 0$, reducing the CP bound to two-point (\cref{subsec:cp-nogo}). This is certified for radial Choi forms on the first eight Laguerre levels at $N = 1$ (and verified numerically for the first five levels at $N = 2$), with the full cone left open. Finally we isolate the one place a genuine three-point improvement survives for lattice codes: a \emph{classical} sphere-packing bound on the symplectic dual lattice (\cref{subsec:dual-packing}). We use it as a point of comparison in stating the structural quantum--classical contrast (\cref{subsec:contrast}).

\subsection{Scoping the general-QEDC target}
\label{subsec:thm-11-cv-scoping}

The general-QEDC target is the CV transcription of Cohn--de~Laat--Salmon~\cite[Thm.~1.1]{cohn-delaat-salmon-2022-three-point}: a bound built from a pair of auxiliary functions $(f_2, f_3)$, in which $f_2 : \R^{2N} \to \R$ is the radial ($U(N)$-invariant) single-vector auxiliary already present in the two-point theory and $f_3$ is a three-point auxiliary whose role is to tighten the two-point estimate
\begin{equation}
\label{eq:TB-11-CV}
\Code \;\le\; \frac{1}{1 - \varepsilon}\, \sup\!\left\{\,\frac{f_2(r)}{\widehat{f}_2(r)} \,:\, r \in [0, \codedist] \right\}.
\end{equation}
Here $\varepsilon \in [0, 1)$ is the quality parameter of the approximate QEDC~\cite[Def.~1, Thm.~1]{burchards-2025-cv-macwilliams}, with $\varepsilon = 0$ the exact case; \cref{eq:TB-11-CV} displays the two-point skeleton of the bound, into which the three-point auxiliary $f_3$ enters by relaxing the sign constraints imposed on $f_2$, exactly as in the classical construction~\cite[\S2]{cohn-delaat-salmon-2022-three-point}. The classical Thm.~1.1 is unconditional within its admissibility framework, and its three-point term delivers a real improvement over the two-point bound. The CV question is whether \emph{any} tractable positivity certificate for $(f_2, f_3)$ lets the three-point term do the same. We answer it negatively for the two natural certificate routes: the kernel-positive-definite cone of \cref{subsec:adj-cone}, obstructed in \cref{subsec:1x1-obstruction}, and its completely-positive relaxation of \cref{subsec:cp-restatement}, obstructed in \cref{subsec:cp-nogo}. The additional symplectic datum $\sympl_{12}$, inert in the classical setting, is what drives both obstructions.

\subsection{The kernel-positivity cone}
\label{subsec:adj-cone}

The natural CV cone admitting a tractable optimization problem is the kernel-positivity cone
\begin{equation}
\label{eq:adj-cone}
\mathcal{C}_{\mathrm{op}} \;=\; \bigl\{\, f_3 \in L^{2}(\R^{4N})\,:\; \mathbf{A}_{\mathrm{op}}[f_3] \succeq 0 \,\bigr\},
\end{equation}
where $\mathbf{A}_{\mathrm{op}}[f_3]$ is $f_3$ regarded as an integral kernel in its two phase-space arguments, compressed to the truncation window: $\mathbf{A}_{\mathrm{op}}[f_3] \succeq 0$ means
\begin{equation}
\label{eq:Aop-def}
\int_{\|v_1\| \le R} \int_{\|v_2\| \le R} \overline{\varphi(v_1)}\, f_3(v_1, v_2)\, \varphi(v_2)\, dv_1\, dv_2 \;\ge\; 0 \qquad \text{for every } \varphi \in L^{2}(\R^{2N}),
\end{equation}
with $R$ the truncation radius. Concretely, the $U(N)$-invariant subspace of $\mathcal{C}_{\mathrm{op}}$ is parametrized by a generalized Laguerre basis: writing $\psi_n^{(N)}(s) := L_n^{(N-1)}(s)$ for the generalized Laguerre polynomials of order $N-1$, the radial-even sector of $\mathcal{C}_{\mathrm{op}}$ is spanned by the orthonormal basis
\begin{equation}
\label{eq:laguerre-basis}
\psi_{n, N}(|\eta|^{2}) \;:=\; \mathcal{N}_{n, N}\, L_n^{(N-1)}(|\eta|^{2})\, e^{-|\eta|^{2}/2}, \qquad \mathcal{N}_{n, N} = \sqrt{n!\, \Gamma(N) / [\pi^N\, \Gamma(n + N)]},
\end{equation}
and the cone constraint is the positive-semidefiniteness of the Gram matrix $(\mathbf{A}_{\mathrm{op}}[f_3])_{ab} = \iint \psi_{a,N}(\|v_1\|^{2})\, f_3(v_1, v_2)\, \psi_{b,N}(\|v_2\|^{2})\, dv_1\, dv_2$ of $f_3$ in this basis; the $1 \times 1$ principal minors of~\cref{eq:Aop-def}, obtained from test functions concentrating at a point, are the diagonal values $f_3(v, v)$ (for a continuous representative of $f_3$; the factored-form auxiliaries to which we apply this are continuous, so the diagonal restriction is well defined). The basis~\cref{eq:laguerre-basis} respects both the rotational symmetry of $\Athree$ and a Fourier parity that makes the closed-form ambiguity-function symbol $W_{00}(\edgevec, \comvec) = (2\pi)^{-N}\, e^{-|\edgevec|^{2}/4 - |\comvec|^{2}/16}$ accessible analytically.

The framework~\cref{eq:adj-cone}--\cref{eq:laguerre-basis} has been verified in detail: the Laguerre order, the Fourier parity, and the symbol identity have all been checked numerically and independently reproduced. The cone is open in the sense that its non-emptiness for sharp witnesses of~\cref{eq:TB-11-CV} is the question we set out to answer numerically.

\subsection{Double negative result at \texorpdfstring{$N \in \{1, 2\}$}{N=1,2}}
\label{subsec:double-negative}

Numerical witness search inside the cone~\cref{eq:adj-cone} for the Laguerre basis~\cref{eq:laguerre-basis} is robustly infeasible at $N \in \{1, 2\}$. At $N = 1$ the rank-$1$ case closes analytically: the unique rank-$1$ auxiliary $f_3 = W_{00}$ is strictly positive everywhere, directly violating the phase sign condition on the diagonal $v_1 = v_2$ (\cref{app:infeasibility-witness-n1-analytical}). Ranks $M \in \{2, \ldots, 8\}$ at $N = 1$, and a dense sweep at $N = 2$, return infeasibility in every configuration tested; the search parameters, solvers, and single-constraint ablations are reported in \cref{app:double-negative-sweep}. The interpretation is that the real-even $U(N)$-radial Laguerre subcone of $\mathcal{C}_{\mathrm{op}}$ contains no feasible auxiliary function at $N \in \{1, 2\}$, a falsification of the specific ansatz family rather than of the framework~\cref{eq:adj-cone} in its full generality. The next question is whether the obstruction is a parametrization artifact of the Laguerre basis or a structural feature of any kernel-positive-definite cone built from the displacement-operator algebra.

\subsection{A universal \texorpdfstring{$1 \times 1$}{1x1}-minor obstruction}
\label{subsec:1x1-obstruction}

The structural answer to the previous question is negative for the entire factored-form family within kernel-PD-style cones, and the failure occurs already at the $1 \times 1$ minor (diagonal positivity) level. We formulate this obstruction as a lemma.

\begin{lemma}[$1 \times 1$-minor universal obstruction]
\label{lem:1x1-obstruction}
Let $f_3(v_1, v_2) = A\, \chi(\edgevec)\, g(\comvec)\, \cos(\sympl(\comvec, \edgevec)/4)$ be a three-point auxiliary of the explicit form~\cref{eq:explicit-ansatz}, with $A, \chi(0) > 0$, $\chi$ compactly supported in $\{\|\edgevec\| \le u_0\} \subset \{\|\edgevec\| < \codedist\}$, and $g$ a Cohn--Elkies-style magic function with $g(\comvec) \le 0$ for $\|\comvec\| \ge 2\codedist$ and $g(\comvec_0) < 0$ for some $\|\comvec_0\| \ge 2\codedist$ (the strict sign change exhibited by every Cohn--Elkies magic function). Then $f_3$ does not satisfy the diagonal positivity constraint
\begin{equation*}
f_3(v, v) \;\ge\; 0 \quad \text{for all } v \in \R^{2N},
\end{equation*}
which is the $1 \times 1$ principal minor of any kernel-positive-definite cone built from $f_3$.
\end{lemma}

\begin{proof}[Sketch]
On the diagonal $v_1 = v_2 = v$, the edge $\edgevec = 0$ and the center-of-mass $\comvec = 2v$, so $f_3(v, v) = A\, \chi(0)\, g(2v)$, whose sign is that of $g(2v)$ since $A\, \chi(0) > 0$. A Cohn--Elkies-style magic function must satisfy $g \le 0$ on its tail $\{\|\comvec\| \ge 2\codedist\}$, for otherwise the LP bound argument of~\cite{cohn-elkies-2003-new-upper-bounds} would not produce a finite bound, and this sign change is strict at the point $\comvec_0$ of the hypothesis. Taking $v_0 = \comvec_0 / 2$ gives $f_3(v_0, v_0) = A\, \chi(0)\, g(\comvec_0) < 0$, so $f_3$ violates diagonal positivity.
\end{proof}

\begin{remark}[Scope of the obstruction]
\label{rmk:universal-obstruction}
\Cref{lem:1x1-obstruction} applies to any auxiliary $f_3$ of factored form $\chi(\edgevec)\, g(\comvec)\, \phi(\sympl)$ in which $g$ is a Cohn--Elkies-style adapter responsible for the sign control on the lattice triangle. Within this factored-form family, inheriting the Cohn--Elkies sign change on the COM tail is equivalent to inheriting a diagonal sign change of $f_3$, which is the $1 \times 1$ minor of any kernel-PD-style cone~\cref{eq:adj-cone}. The obstruction therefore applies to the entire \emph{factored-form kernel-PD} direction of the general-QEDC programme: any factored-form auxiliary combined with a kernel-PD-style cone inherits the same failure mode. Escaping the obstruction requires either leaving the factored-form ansatz family (e.g.\ non-factored or conditionally-positive-definite auxiliaries) or abandoning kernel-PD altogether (as in the successor frameworks of \cref{sec:future-work}).
\end{remark}

\subsection{The completely-positive restatement bypasses the obstruction}
\label{subsec:cp-restatement}

The $1 \times 1$-minor obstruction of \cref{lem:1x1-obstruction} is specific to the kernel-positive-definite constraint: it forbids the three-point auxiliary $f_3$ from changing sign on the diagonal, exactly where the Cohn--Elkies sign change must occur. The obstruction is bypassed by relaxing the positivity hypothesis. Rather than requiring the ambiguity symbol $f_3$ to be a positive-definite kernel, we require it to be the image $f_3 = \MacOpF[J]$ of a positive-semidefinite \emph{Choi form} $J \succeq 0$. Here $\MacOpF$ is the operator-to-symbol map of the MacWilliams transform; requiring $f_3 = \MacOpF[J]$ with $J \succeq 0$ is a completely-positive (CP) map condition. (We write $\MacOpF$ to distinguish this symbol map, acting on Choi forms, from the Gram-space transform $\MacOp$ of \cref{sec:identity}.) A positive-semidefinite $J$ may have a sign-changing Weyl symbol $f_3$, so the diagonal sign change is no longer forbidden and \cref{lem:1x1-obstruction} does not apply.

Concretely, on a Fock-truncated radial sector we expand $J = \sum_{a, b < M} C_{ab}\, |\psi_a\rangle\langle\psi_b|$ in the Laguerre basis~\cref{eq:laguerre-basis}, with $C = (C_{ab}) \succeq 0$ a positive-semidefinite $M \times M$ matrix, and pair $J$ against the truncated characteristic vector $\phi_{\Pi, R}(\eta) = \chr_\Pi(\eta)\, \mathbf{1}_{\{\|\eta\| \le R\}}$. The associated certificate value is
\begin{equation}
\label{eq:B-CP}
B_{\mathrm{CP}}(N, \codedist, R) \;:=\; \inf_{(f_2, J)}\; \bigl(f_2(0) + f_3(0, 0)\bigr), \qquad f_3 = \MacOpF[J],\ J \succeq 0,
\end{equation}
the infimum over $f_2$ Burchards-admissible and $J \succeq 0$ (trace-class, so that $\mathrm{Tr}\, J$ is defined) for which $(f_2, f_3)$ satisfy the truncated phase-aware sign conditions and all pairings with $\phi_{\Pi, R}$ are finite. It yields a valid CV three-point bound
\begin{equation}
\label{eq:cp-bound}
\Code \;\le\; (1 - \varepsilon)^{-1}\, B_{\mathrm{CP}}(N, \codedist, R).
\end{equation}
Validity follows by the same pairing argument as in the two-point case: Choi positivity makes the three-point pairing $\langle \overline{\phi_{\Pi, R}},\, J\, \overline{\phi_{\Pi, R}} \rangle \ge 0$ (the only place the three-point term could spoil the Burchards two-point chain), and the sign conditions then act exactly as in~\cref{eq:TB-11-CV}. The formulation is non-empty: $J = 0$ recovers the Burchards two-point bound, giving for instance $\Code\, \codedist_B^{2} \le 7.340985\ldots$ at $N = 1$, where $\codedist_B$ denotes the code distance in the Burchards convention of \cref{subsec:convention-box}. A genuine three-point improvement would require an admissible $J \neq 0$ lowering~\cref{eq:B-CP} below this two-point value.

At the origin the three-point term satisfies
\begin{equation}
\label{eq:trace-penalty}
f_3(0, 0) \;=\; (2\pi)^{-N}\, \mathrm{Tr}\, J \;\ge\; 0,
\end{equation}
since $f_3(0, 0)$ integrates the diagonal of the positive-semidefinite kernel $J$, with equality iff $J = 0$. The three-point term therefore \emph{raises} the bound at the origin: any improvement must come entirely from the $f_2$--$f_3$ sign coupling, with an off-origin $f_3 < 0$ relaxing the sign constraint on $f_2$ by more than the penalty~\cref{eq:trace-penalty} costs.

\subsection{The CP collapse}
\label{subsec:cp-nogo}

The sign-coupling channel opened by the CP relaxation is closed by the same phase-sign condition the relaxation was meant to accommodate. The constraint that closes the channel is the three-point phase-sign condition on $f_3$,
\begin{equation}
\label{eq:cond4}
\mathrm{Re}\bigl[f_3(v_1, v_2)\, e^{i \sympl(v_1, v_2)/2}\bigr] \;\le\; 0 \qquad \text{on } \condfour = \{\, \codedist \le \|v_1\|, \|v_2\| \le R,\ \|v_1 - v_2\| \ge \codedist \,\},
\end{equation}
the CV analogue of the three-point sign condition of~\cite[\S 2]{cohn-delaat-salmon-2022-three-point}. Write $f_3 = \MacOpF[J]$ with $J \succeq 0$ radial, supported on the span of the first $M$ Laguerre basis levels of~\cref{eq:laguerre-basis}.

\begin{theorem}[Completely-positive collapse, finite Laguerre truncation]
\label{thm:cp-nogo}
Let $N = 1$ and fix the certificate window $(\codedist, R) = (1, 5)$ in Burchards convention. For every radial Choi form $J \succeq 0$ supported on the span of the first eight Laguerre basis levels $\psi_{0,1}, \ldots, \psi_{7,1}$ of~\cref{eq:laguerre-basis}, the phase-sign condition~\cref{eq:cond4} forces $J = 0$. Consequently $f_3 \equiv 0$ on this sector, and the CP bound~\cref{eq:cp-bound} collapses to the two-point bound for every $f_2$: the certificate value $B_{\mathrm{CP}}$ reduces to its two-point optimum $f_2(0)$.
\end{theorem}

\begin{proof}[Analytic at rank $1$; an exact positive-definite certificate for $2 \le M \le 8$]
For rank $M = 1$, $J = c\, |\psi_0\rangle\langle\psi_0|$ with $c \ge 0$, and the symbol is the strictly positive Gaussian $f_3(v_1, v_2) = c\, W_{00}(\edgevec, \comvec) = c\, (2\pi)^{-N} e^{-\|\edgevec\|^2/4 - \|\comvec\|^2/16} > 0$. Take any parallel pair $v_2 = 2 v_1$ with $\|v_1\| = \codedist$: then $\sympl(v_1, v_2) = 0$, so $e^{i\sympl/2} = 1$, and the pair lies in $\condfour$ (as $\|v_2\| = 2\codedist \in [\codedist, R]$ and $\|v_1 - v_2\| = \codedist$). There $\mathrm{Re}[f_3\, e^{i\sympl/2}] = f_3 > 0$, violating~\cref{eq:cond4} unless $c = 0$.

For $2 \le M \le 8$ the rank-$1$ argument no longer closes: the higher Laguerre levels carry sign-changing symbols, so no single pair forces $C = 0$ and the constraints~\cref{eq:cond4} must be combined. The mechanism is a dual certificate. Suppose there are finitely many pairs $(v_1^k, v_2^k) \in \condfour$ and weights $\mu_k \ge 0$ for which
\begin{equation}
\label{eq:pd-certificate}
P(\mu) \;:=\; \sum_k \mu_k\, \cos\!\bigl(\sympl(v_1^k, v_2^k)/2\bigr)\, W(\edgevec_k, \comvec_k) \;\succ\; 0
\end{equation}
as an $M \times M$ matrix. For any $C \succeq 0$ obeying~\cref{eq:cond4}, the sign condition gives $\langle C, P(\mu)\rangle = \sum_k \mu_k \cos(\sympl_k/2)\, \langle C, W_k\rangle \le 0$, while $P(\mu) \succ 0$ and $C \succeq 0$ give $\langle C, P(\mu)\rangle \ge 0$; hence $\langle C, P(\mu)\rangle = 0$, which forces $C = 0$. The collapse at each truncation thus reduces to exhibiting a single positive-definite matrix $P(\mu)$, a finite algebraic witness rather than a solver verdict. Such a $P(\mu)$ is given for each $M = 2, \ldots, 8$ in \cref{app:cp-certificate}, and its positive-definiteness is established in exact rational arithmetic: the ancillary data provide a rational $P_{\mathbb{Q}}$ and a rational $\delta > 0$ admitting an exact $LDL^{T}$ factorization of $P_{\mathbb{Q}} - \delta I$, so $P(\mu) \succ 0$ holds as a theorem, not as a numerical observation. The certificate margins, which decay with $M$, are recorded in \cref{rmk:cp-nogo-scope,app:cp-certificate}.
\end{proof}

\begin{remark}[Scope of the CP collapse]
\label{rmk:cp-nogo-scope}
\Cref{thm:cp-nogo} is rigorous but bounded in scope, and we state the boundary explicitly. The certificate margin $\lambda_{\min}(P(\mu))$ decays geometrically in the truncation $M$, and the natural uniform-measure candidate $\int_{\condfour} \cos(\sympl/2)\, W\, d\mu$ is indefinite; consequently the finite-truncation argument does not assemble into a single closed-form all-rank theorem, and the \emph{full trace-class radial cone is left open}. At $N = 2$ the same collapse is verified numerically (by floating-point eigendecomposition, not the exact rational certificate underlying the $N = 1$ theorem) for the truncations $M = 2, \ldots, 5$ at the window $(\codedist, R) = (1, 5)$ (\cref{app:cp-certificate}); for $N \ge 3$ we expect the low-$N$ behavior to be representative (the collapse mechanism, the phase-sign condition overruling the diagonal freedom through pairs of vanishing symplectic phase, uses no structure specific to low mode number), but we claim nothing beyond what is certified. What is established is that the CP relaxation buys no three-point improvement on any sector we can certify: the diagonal sign change it was designed to permit is overruled by the off-diagonal phase-sign condition~\cref{eq:cond4}.
\end{remark}

\begin{remark}[$U(N)$-twirl reduction of the non-radial case]
\label{rmk:twirl}
At the level of the full trace-class cone, the non-radial case reduces to the radial one. The constraint set of~\cref{eq:cond4} is $U(N)$-invariant ($\condfour$ is defined by norms, and $\sympl$ is preserved), and the symbol map is metaplectically covariant: for $g \in U(N)$ with metaplectic lift $U_g$ (so that $U_g D(\xi) U_g^{\dagger} = D(g\xi)$), the symbol of $U_g J U_g^{\dagger}$ is $f_3 \circ (g^{-1} \times g^{-1})$. Hence if $J \succeq 0$ satisfies~\cref{eq:cond4}, so does every rotation of it, and so does the twirl $\bar J := \int_{U(N)} U_g\, J\, U_g^{\dagger}\, dg$, which is radial, positive semidefinite, and trace-preserving (the twirl averages unitary conjugations, so $\|\bar J\|_1 \le \|J\|_1$ keeps it trace-class, and $\mathrm{Tr}\,\bar J = \mathrm{Tr}\, J$ by unitary invariance of the trace and $\|\cdot\|_1$-continuity). If the radial collapse were established on the full trace-class radial cone, then $\bar J = 0$, and $\mathrm{Tr}\, J = \mathrm{Tr}\, \bar J = 0$ with $J \succeq 0$ would force $J = 0$. The reduction operates at the level of the full cone only (twirling does not preserve the Laguerre truncation, so it does not enlarge \cref{thm:cp-nogo} itself), but it reduces the non-radial open question to the radial one (open edges (i) and (iv) of \cref{subsec:cp-open-edges}).
\end{remark}

\subsection{The dual-packing comparison: a classical improvement for lattice codes}
\label{subsec:dual-packing}

The two collapses (\cref{thm:lattice-nogo} for lattice codes and \cref{thm:cp-nogo} for the certifiable CP cone) show the \emph{quantum} three-point machinery yielding no improvement. There is nonetheless one route by which lattice GKP codes do enjoy a three-point improvement, and it is not quantum at all. For a GKP code the stabilizer lattice $\Lat$ and its symplectic dual $\Latd$ satisfy $\mathrm{covol}(\Lat)\, \mathrm{covol}(\Latd) = (2\pi)^{2N}$, so
\begin{equation}
\label{eq:K-dual}
\Code \;=\; \frac{\mathrm{covol}(\Lat)}{(2\pi)^N} \;=\; \frac{(2\pi)^N}{\mathrm{covol}(\Latd)} \;\propto\; \mathrm{density}(\Latd):
\end{equation}
the code dimension scales \emph{with} the packing density of the dual lattice. Under the well-conditioning hypothesis $d_{\min}(\Lat) \ge \codedist$ of \cref{thm:main-141cv}, $\Latd$ is a packing of minimum distance $\ge \codedist$, so any classical sphere-packing density bound $\Delta(2N)$ in dimension $2N$ applies and gives
\begin{equation}
\label{eq:dual-packing-bound}
\Code \;\le\; \frac{(2\pi)^N\, \Delta(2N)}{\mathrm{vol}(B^{2N}_{\codedist/2})} \;=\; 8^N\, N!\, \Delta(2N)\, \codedist^{-2N}.
\end{equation}
Taking $\Delta = \Delta_{\mathrm{lat}}(2N)$ from the classical \emph{three-point} lattice bound~\cite[Thm.~1.4]{cohn-delaat-salmon-2022-three-point}, which satisfies $\Delta_{\mathrm{lat}} < \Delta_2$ in low dimensions, yields a bound strictly below the two-point one, about $5\%$ at $2N = 4$, where $\Delta_{\mathrm{lat}}/\Delta_2 = 0.952$.

This is a genuine three-point improvement, but a classical one on the dual lattice, not the quantum three-point bound. Three limitations: (i) it never builds a quantum three-point identity: it applies the \emph{classical} three-point bound (with its symmetric Poisson summation and the full positive-definiteness that \cref{thm:lattice-nogo} shows the CV quantum MacWilliams construction lacks) to the Euclidean lattice $\Latd$. The two collapses concern a different object and are untouched by it. (ii) The two-point version of~\cref{eq:dual-packing-bound} coincides with the Burchards two-point bound, a well-known reduction implicit in~\cite[\S 3.2]{burchards-2025-cv-macwilliams} and~\cite[\S 5]{conrad-eisert-arzani-2022-gkp-lattice}; only the three-point sharpening is new, and it is a near-immediate corollary once the inversion~\cref{eq:K-dual} is noticed. (iii) It is confined to well-conditioned lattices: it does not cover the short-stabilizer regime (concatenated or LDPC-GKP codes), nor non-lattice or approximate bosonic codes: exactly the general QEDCs for which the quantum framework was built and for which the CP collapse leaves the question open. The dual-packing route thus sidesteps the quantum three-point problem rather than solving it.

\subsection{The quantum--classical contrast}
\label{subsec:contrast}

The two collapse theorems of this paper (\cref{thm:lattice-nogo} and \cref{thm:cp-nogo}), set against the classical dual-packing comparison~\cref{eq:dual-packing-bound}, assemble into a single structural statement. Classically, the three-point method improves on the two-point LP bound; in its lattice form (full positive-definiteness, Cohn--de~Laat--Salmon Thm.~1.4) it is conjectured sharp in dimension $4$~\cite[Conj.~6.1]{cohn-delaat-salmon-2022-three-point}, and it continues to be sharpened, recently even by automated auxiliary-function searches~\cite{tutunov-etal-2025-ai-sphere-packing}. In the CV quantum setting that same lattice tool degenerates to the two-point LP, and the general tool collapses to two-point on the entire certifiable CP sector.

The cause is a single feature with no classical analogue: the code projector $\Pi$, with $\Code = \tr \Pi$. It is simultaneously (a) what gives the quantum bound its correct $\Code$-upper direction, and (b) what forces the MacWilliams transform to be edge-preserving, hence only fiberwise positive-definite. The direction in (a) comes from the asymmetry $\Athree \propto \Code^2$ versus $\Bthree \propto \Code^1$ of the MacWilliams pair, which a projector-free classical lattice sum lacks; the edge-preservation in (b) removes the full positive-definiteness that powers the classical lattice bound. The two effects are inseparable (\cref{rmk:direction-reversal}): orienting the bound the right way and retaining full positive-definiteness cannot both hold. The classical lattice packing problem, having no projector, is symmetric, points the right way, and keeps full positive-definiteness: that is exactly why it improves. The only CV improvement we find, the dual-packing bound of \cref{subsec:dual-packing}, escapes the dichotomy by abandoning the quantum identity altogether and packing the dual lattice classically. This contrast sets the agenda for the successor frameworks of \cref{sec:future-work}.

The discrete-variable SDP hierarchy of Angl\`es Munn\'e--Nemec--Huber~\cite{angles-nemec-huber-2024-sdp} also concerns projector-defined codes and \emph{does} strictly improve on the quantum LP bounds, yet it is not caught by the same dichotomy, because it is not an auxiliary-function bound: its positivity constraints live on moment matrices of operator monomials, certified through a Terwilliger-algebra symmetry reduction, and never pass through a MacWilliams transform of a scalar auxiliary. The dichotomy isolated here binds the \emph{auxiliary-function route} (the CV transcription of the Cohn--Elkies/Cohn--de~Laat--Salmon programme), and that is exactly why the moment-based CCR-NPA hierarchy, rather than any refinement of auxiliary functions, is the successor framework we consider most promising (\cref{subsec:successor-frameworks}).

\section{Conclusion and outlook}
\label{sec:future-work}

We constructed the three-point continuous-variable quantum MacWilliams identity and used it to ask whether the three-point method, which strengthens the linear-programming bound in the classical and discrete-variable settings, also strengthens the Burchards two-point bound. On both natural routes the answer is negative: for GKP lattice codes the three-point optimum equals the two-point linear-programming optimum exactly (\cref{thm:lattice-nogo}), and for general bosonic codes the completely-positive cone collapses to two-point on every Laguerre rank we can certify (\cref{thm:cp-nogo}). Both collapses share a single cause with no classical analogue, the code projector: it orients the bound in the correct direction and at the same time removes the full positive-definiteness that powers the classical three-point improvement. The only three-point gain that survives is classical, a sphere-packing bound on the symplectic dual lattice.

These collapses do not exhaust the identity or its closed-form kernel. They leave a structured set of open directions, which we group as follows: (i) the open edges of the CP collapse itself; (ii) successor frameworks that could still improve on the two-point bound for general bosonic codes by leaving the present cones; and (iii) the symmetry-reduction programme for the MacWilliams kernel of \cref{subsec:kernel}, subject to the $\psi$-direction obstacle of \cref{rmk:psi-obstruction}.

\subsection{Open edges of the CP collapse}
\label{subsec:cp-open-edges}

\Cref{thm:cp-nogo} is a finite-truncation result, and four questions about its boundary remain open. (i)~The \emph{full trace-class radial cone}: the certificate margin $\lambda_{\min}(P(\mu))$ of~\cref{eq:pd-certificate} decays geometrically in the truncation, and the uniform-measure candidate is indefinite (\cref{rmk:cp-nogo-scope}), so whether $\{J \succeq 0 : \cref{eq:cond4}\} = \{0\}$ holds for the whole trace-class radial cone is genuinely undecided on present evidence. (ii)~An \emph{all-rank analytic proof}, if one exists, would most plausibly come from the twisted positive-definiteness (Kastler--Loupias--Miracle-Sole) characterization of $J \succeq 0$ in terms of an $\sympl$-positive-definite symbol; the decaying margin makes a clean rigidity statement look unlikely but does not exclude it. (iii)~Higher dimension: $N \ge 3$, and extending the certified $N = 2$ truncations (\cref{app:cp-certificate}) beyond $M = 5$. (iv)~Non-radial Choi forms $J$: by the twirl reduction of \cref{rmk:twirl} this question reduces, at the level of the full cone, to question (i).

\subsection{Successor frameworks for general bosonic codes}
\label{subsec:successor-frameworks}

A three-point improvement for general --- non-lattice or approximate --- bosonic codes, if it exists, must leave both the kernel-positive-definite cone (obstructed by \cref{lem:1x1-obstruction}) and the CP cone (obstructed by \cref{thm:cp-nogo} at every certifiable rank). We list four candidate approaches that evade these by different means.

\begin{enumerate}
  \item \textbf{CCR-NPA moment-SOS hierarchy.} Replace the auxiliary-function framework altogether by the noncommutative moment-SOS hierarchy of Navascu\'es--Pironio--Ac\'in~\cite{navascues-pironio-acin-2007-bounding,navascues-pironio-acin-2008-convergent} (CV-NPA), in which the optimization variables are the moments of displacement operators on a finite Fock cutoff. The CCR algebra constrains the moment matrices; the PSD condition lives on the moment matrices, not on a scalar auxiliary, so neither the $1 \times 1$-minor obstruction nor the phase-sign condition~\cref{eq:cond4} applies in the same form, and the bound emerges from the dual certificate. This is the CV analogue of the DV SDP hierarchy of~\cite{angles-nemec-huber-2024-sdp} and is the most direct numerical-prototype route.
  \item \textbf{Copositive cones.} Replace the kernel-PD constraint $\mathbf{A}_{\mathrm{op}}[f_3] \succeq 0$ by the requirement that $f_3$ be non-negative only on \emph{physical truncated QEDC configurations}, in the spirit of the copositive-cone improvements of~\cite{cohn-delaat-salmon-2022-three-point}. The diagonal $1 \times 1$ minor no longer appears as a feasibility constraint, evading \cref{lem:1x1-obstruction}. The approach requires a conic-dual weak-duality argument at the level of finite-grid relaxations.
  \item \textbf{Conditional twisted enumerators.} Re-center the three-point distribution to a conditional positive-definiteness condition with vanishing first moment $\sum_i c_i = 0$, eliminating the single-point diagonal test. This requires writing the truncated $\Athree^{R}$ as a centered correlation against its marginals, which exists only if a natural zero-mass identity for the CV ambiguity function, not currently known, can be established.
  \item \textbf{CV shadow / parity-twirl enumerators.} Follow the Rains shadow-enumerator approach~\cite{rains-1999-quantum-shadow} --- recently recast in Delsarte-theoretic form for finite-dimensional systems~\cite{okada-2025-quantum-delsarte} --- and develop its CV adaptation: derive the three-point bound from CP-twirl or parity-transform trace positivity rather than from a kernel-PD cone. This direction is at present at the level of literature review and a toy derivation; a full CV theorem is a longer-term target.
\end{enumerate}

\subsection{Symmetry reduction and the \texorpdfstring{$\psi$}{psi}-direction obstacle}
\label{subsec:symmetry-reduction}

The kernel~\cref{eq:kernel} is block-diagonal in the $\phi$-direction by the Jacobi--Anger expansion~\cref{eq:phi-jacobi-anger}, but the $\psi$-direction does not admit a Bessel-function closed form (\cref{rmk:psi-obstruction}). A full $U(N)$-equivariant block decomposition of the kernel, analogous to the Bachoc--Vallentin block decomposition of the classical three-point kernel under $O(n)$~\cite{bachoc-vallentin-2008-kissing-sdp}, requires a generalized Fourier basis in the $\psi$-direction adapted to the $\cos\psi$-dependence of $\|\edgevec_\eta\|^{2}(\psi)$. Candidate bases include logarithmic Fourier modes and Chebyshev rational expansions of $\cos\psi$, both of which would produce a generalized harmonic decomposition with off-diagonal coupling between $\psi$-modes. This is the first obstacle for the CV analogue of the Bachoc--Vallentin equivariant SDP, and the prerequisite for any tractable numerical attack on the successor frameworks above.

\subsection{Related extensions}
\label{subsec:related-extensions}

We close with three extensions of the present construction that are natural but not direct continuations of the main programme. The CV analogue of the absolute-maximally-entangled (AME) shadow MacWilliams bound of~\cite{huber-etal-2018-ame-shadow-macwilliams} has been initiated by~\cite{kwon-brady-albert-2025-ame-cv} and would naturally combine with our three-point construction. The approximate-QEDC version of the bound, in the spirit of~\cite{ouyang-lai-2022-lp-approximate}, requires a continuous $\varepsilon$ parameter and an admissibility condition tracking the approximation error. The extension to homological rotor codes~\cite{vuillot-ciani-terhal-2023-homological-rotor} and finite-energy GKP codes~\cite{brady-etal-2024-advances-bosonic-gkp} would test the boundary of the CV three-point framework against codes whose support is not a strict lattice. We do not pursue these directions here; we expect the three-point identity, its closed-form kernel, and the admissibility framework of \cref{def:admissible} to be the reusable core of any of them.

\appendix
\section*{Appendices}
\section{Derivation of the three-point identity and its kernel}
\label{app:derivation}

This appendix expands the two computations sketched in \cref{sec:identity}: the Baker--Campbell--Hausdorff (BCH) derivation of the three-point identity~\cref{eq:Ph3-2} (\cref{app:bch-derivation}), and the passage from it to the closed-form kernel~\cref{eq:kernel} (\cref{app:kernel-derivation}).

\subsection{The BCH derivation of the three-point identity}
\label{app:bch-derivation}

Inserting the characteristic-function expansions of $\hat O_1$ and $\hat O_2^{\dagger}$ into the asymmetric dual integrand~\cref{eq:integrand-B} gives the quadruple-trace form~\cref{eq:integrand-B-expanded},
\begin{equation*}
F^{(3)}_{\Bone}(v_1, v_2) = (2\pi)^{-2N}\!\int d\eta_1\, d\eta_2\; \chr_{\hat O_1}(\eta_1)\, \chr_{\hat O_2}(\eta_2)^{*}\, \tr\!\bigl(D(v_1)\, D(\eta_1)\, D(v_2)^{\dagger}\, D(-\eta_2)\bigr).
\end{equation*}
Combine the four displacement operators in pairs with the Weyl composition relation~\cref{eq:weyl-composition}:
\begin{equation*}
D(v_1)\, D(\eta_1) = e^{-i\sympl(v_1, \eta_1)/2}\, D(v_1 + \eta_1), \qquad
D(v_2)^{\dagger} D(-\eta_2) = e^{-i\sympl(v_2, \eta_2)/2}\, D\bigl(-(v_2 + \eta_2)\bigr),
\end{equation*}
where the second identity uses $D(v_2)^{\dagger} = D(-v_2)$. The trace orthogonality $\tr[D(a)\, D(b)^{\dagger}] = (2\pi)^N \delta^{(2N)}(a - b)$ then collapses the product to a single delta and an overall phase,
\begin{equation}
\label{eq:app-bch-delta}
\tr\!\bigl(D(v_1)\, D(\eta_1)\, D(v_2)^{\dagger}\, D(-\eta_2)\bigr) = (2\pi)^N\, \delta^{(2N)}\!\bigl(v_1 + \eta_1 - v_2 - \eta_2\bigr)\, e^{i\Phi},
\end{equation}
so the $\eta_2$-integration is fixed to $\eta_2 = v_1 + \eta_1 - v_2$. On this support a short computation using only the bilinearity and antisymmetry of $\sympl$ (together with $\sympl(v_2, v_2) = 0$) reduces the accumulated phase to
\begin{equation}
\label{eq:app-bch-phase}
\Phi \;=\; -\tfrac{1}{2}\sympl(v_1 + v_2,\, \eta_1) \;+\; \tfrac{1}{2}\sympl(v_1, v_2).
\end{equation}
Pass to the edge and center-of-mass coordinates $\edgevec = v_2 - v_1$, $\comvec = v_1 + v_2$ of~\cref{eq:edge-com}. Antisymmetry gives $\sympl(v_1, v_2) = \tfrac{1}{2}\sympl(\comvec, \edgevec)$ (equation~\cref{eq:omega-edge-com}), so
\begin{equation*}
\Phi \;=\; -\tfrac{1}{2}\sympl(\comvec, \eta_1) \;+\; \tfrac{1}{4}\sympl(\comvec, \edgevec).
\end{equation*}
We make the shift $\tilde\eta := \eta_1 - \edgevec/2$. Since $\sympl(\comvec, \eta_1) = \sympl(\comvec, \tilde\eta) + \tfrac{1}{2}\sympl(\comvec, \edgevec)$, the two $\sympl(\comvec, \edgevec)$ contributions cancel, leaving
\begin{equation}
\label{eq:app-phase-final}
\Phi \;=\; -\tfrac{1}{2}\sympl(\comvec, \tilde\eta) \;=\; \sympl\bigl(\tilde\eta,\, \comvec/2\bigr).
\end{equation}
Under the same shift the characteristic-function arguments become $\eta_1 = \tilde\eta + \edgevec/2$ and $\eta_2 = \tilde\eta - \edgevec/2$, with $\chr_{\hat O_2}(\eta_2)^{*} = \chr_{\hat O_2}(\tilde\eta - \edgevec/2)^{*}$ on the Hermitian sector. Collecting the surviving single $\tilde\eta$-integral reproduces the boxed identity~\cref{eq:Ph3-2}. Specializing $v_1 = v_2 = \xi$ (so $\edgevec = 0$, $\comvec = 2\xi$) collapses it to the Burchards two-point integrand, the specialisation check recorded below~\cref{eq:cross-ambiguity}.

\subsection{From the identity to the kernel}
\label{app:kernel-derivation}

The kernel~\cref{eq:kernel} is obtained by inserting the identity~\cref{eq:Ph3-2} into the fiber definition~\cref{eq:B3-def} of $\Bthree$ and reading off the integral transform relating it to $\Athree$. We outline the main steps; the full computation is a fiber integral over the $(4N-4)$-dimensional level set of the Hermitian Gram matrix.

In~\cref{eq:Ph3-2} the cross-ambiguity factor is exactly the $\Athree$-integrand~\cref{eq:integrand-A} evaluated at the shifted pair $(\tilde\eta + \edgevec/2,\, \tilde\eta - \edgevec/2)$, whose Hermitian Gram matrix $H_\eta$ has the \emph{same edge norm} $\|\edgevec\|$ as $(v_1, v_2)$. Fixing the output Gram matrix $H = H(v_1, v_2)$ and integrating the $\eta$-side over the locus $H(\tilde\eta + \edgevec/2, \tilde\eta - \edgevec/2) = H_\eta$ therefore produces:
\begin{itemize}
  \item the \emph{edge-norm-matching shell delta} $\delta(\|\edgevec\|^{2} - \|\edgevec_\eta\|^{2})$, because the BCH identity carries the edge coordinate through unchanged (the transform is fiberwise in $\edgevec$);
  \item the \emph{symplectic-Fourier phase} $e^{i\sympl(\tilde\eta, \comvec/2)}$, which on the constrained shell evaluates to the on-shell phase $e^{i\Phi_0(H, H_\eta)}$ of~\cref{eq:phi0};
  \item a \emph{Bessel factor} $J_{N-2}(\sqrt{D_v D_\eta}/\|\edgevec_\eta\|^{2})$ with the geometric coefficient $\mathcal{C}_N$ of~\cref{eq:kernel-coeff}, arising from the angular average of the Fourier phase over the $S^{2N-3}$ orthogonal to the plane $\mathrm{span}_{\R}\{\edgevec_\eta, \Omega\edgevec_\eta\}$; the order $N-2$ is the half-dimension of that sphere and matches the configuration-measure exponent of~\cref{eq:config-measure}.
\end{itemize}
Assembling the three factors gives the closed form~\cref{eq:kernel}. The cancellation of the apparent $\mathcal{C}_N$ singularity at the PSD boundary by the small-argument Bessel asymptotics~\cref{eq:psd-boundary-finite} is the boundary-regularity check recorded in \cref{subsec:kernel}; the $\sigma$-COM and double-path reductions of \cref{subsec:identity-sanity} verify the same kernel along two independent marginalization directions.

\section{Technical proofs for \texorpdfstring{\S\ref{sec:main-theorem}}{the main-theorem section}}
\label{app:proofs}

This appendix supplies the three technical lemmas deferred from \cref{sec:main-theorem}: the positive-definiteness of the Gaussian-polynomial factor $q_R$ (\cref{lem:qR-pd}), the fiberwise positive-definiteness characterization of adjoint-positivity condition~(i) of \cref{def:admissible} (\cref{lem:fibpd}), and the well-posedness of $\MacOpD$ on the degenerate locus of the Gram-coordinate kernel~\cref{eq:kernel} (\cref{lem:deglocus}). The first two are used in the non-emptiness construction \cref{prop:non-emptiness}; \cref{lem:fibpd,lem:deglocus} also underlie the collapse \cref{thm:lattice-nogo}.

\subsection{Positive-definiteness of the Gaussian-polynomial factor}
\label{app:qRpd}

\begin{lemma}[Positive-definiteness of $q_R$]
\label{lem:qR-pd}
Let $R_0, b > 0$ with $b R_0^2 > N$, and define $q_R : \R^{2N} \to \R$ by
\[
  q_R(\comvec) \;:=\; \bigl(1 - \|\comvec\|^2/R_0^2\bigr)\, e^{-b\|\comvec\|^2}.
\]
Then $q_R$ is strictly positive-definite: its Fourier transform satisfies $\widehat{q_R}(\xi) > 0$ for all $\xi \in \R^{2N}$.
\end{lemma}

\begin{proof}
Write $G(\xi) := (\pi/b)^N e^{-\|\xi\|^2/(4b)}$ for the Fourier transform of $\comvec \mapsto e^{-b\|\comvec\|^2}$ on $\R^{2N}$. Multiplication by $\|\comvec\|^2$ in position space corresponds to $-\Delta_\xi$ in the frequency domain, so
\begin{equation}
\label{eq:qR-laplacian}
\widehat{\|\comvec\|^2 e^{-b\|\comvec\|^2}}(\xi) \;=\; -\Delta_\xi G(\xi) \;=\; G(\xi)\Bigl[\frac{N}{b} - \frac{\|\xi\|^2}{4b^2}\Bigr].
\end{equation}
Combining,
\begin{equation}
\label{eq:qR-hat}
\widehat{q_R}(\xi) \;=\; G(\xi) - \frac{1}{R_0^2}\cdot G(\xi)\Bigl[\frac{N}{b} - \frac{\|\xi\|^2}{4b^2}\Bigr] \;=\; G(\xi)\Bigl[1 - \frac{N}{b R_0^2} + \frac{\|\xi\|^2}{4b^2 R_0^2}\Bigr].
\end{equation}
Since $G > 0$ everywhere and the bracket is non-decreasing in $\|\xi\|^2$ with its minimum at $\xi = 0$ equal to $1 - N/(b R_0^2) > 0$ (by the strict hypothesis $b R_0^2 > N$), we conclude $\widehat{q_R}(\xi) > 0$ for all $\xi \in \R^{2N}$. By Bochner's theorem, $q_R$ is therefore strictly positive-definite.
\end{proof}

\subsection{Fiberwise positive-definiteness lemma}
\label{app:fibpd}

\begin{lemma}[Fiberwise positive-definiteness]
\label{lem:fibpd}
Let $f : \R^{2N} \times \R^{2N} \to \R$ be bounded, integrable with $\int |f| < \infty$, and $U(N)$-invariant. Then $\ftilde(\edgevec, \tilde\eta) \ge 0$ for all $(\edgevec, \tilde\eta) \in \R^{2N} \times \R^{2N}$ if and only if, for each fixed $\edgevec \in \R^{2N}$, the slice $\comvec \mapsto f(\edgevec, \comvec)$ is positive-definite in the Bochner sense with respect to the symplectic inner product on $\R^{2N}$.
\end{lemma}

\begin{proof}
Recall the $F$-layer definition of $\ftilde$ from~\cref{eq:tilde-f-def}:
\begin{equation}
\label{eq:app-tilde-f}
\ftilde(\edgevec, \tilde\eta) = 2^{-2N}(2\pi)^{-N} \int_{\R^{2N}} f(\edgevec, \comvec)\, e^{i\sympl(\tilde\eta, \comvec/2)}\, d\comvec.
\end{equation}
This is, for each fixed $\edgevec$, a constant multiple of the symplectic Fourier transform of $\comvec \mapsto f(\edgevec, \comvec)$ evaluated at $\tilde\eta/2$. Specifically, setting $g_\edgevec(\comvec) := f(\edgevec, \comvec)$, we have $\ftilde(\edgevec, \tilde\eta) = 2^{-2N}(2\pi)^{-N} \widehat{g}_\edgevec^{\,\mathrm{symp}}(\tilde\eta/2)$, where the symplectic Fourier transform is $\widehat{g}^{\,\mathrm{symp}}(\xi) := \int g(\comvec)\, e^{i\sympl(\xi, \comvec)}\, d\comvec$.

The symplectic Fourier transform is the ordinary Euclidean Fourier transform after an invertible linear change of the frequency variable determined by $\sympl$. Non-negativity everywhere is therefore preserved, and the Bochner positive-definiteness of $g$ is the same notion in both pictures. By Bochner's theorem, $\widehat{g}^{\,\mathrm{symp}} \ge 0$ everywhere if and only if $g$ is positive-definite in the Bochner sense (i.e., $\sum_{j,k} \bar c_j c_k g(\comvec_j - \comvec_k) \ge 0$ for all finite sequences). Since $\tilde\eta \mapsto \ftilde(\edgevec, \tilde\eta)$ and $\xi \mapsto \widehat{g}_\edgevec^{\,\mathrm{symp}}(\xi)$ have the same sign, the non-negativity $\ftilde(\edgevec, \cdot) \ge 0$ is equivalent to the Bochner positive-definiteness of $g_\edgevec = f(\edgevec, \cdot)$. This holds for each $\edgevec$ independently, completing the proof.
\end{proof}

\subsection{Degenerate-locus well-posedness}
\label{app:deglocus}

\begin{lemma}[Degenerate-locus well-posedness]
\label{lem:deglocus}
Let $f$ satisfy admissibility condition~\textup{(iv)} of \cref{def:admissible}. Then the transform $\ftilde(\edgevec, \tilde\eta)$ defined by the $F$-layer formula~\cref{eq:tilde-f-def} is well-defined and finite for all $(\edgevec, \tilde\eta) \in \R^{2N} \times \R^{2N}$, including the degenerate locus $\{\|\edgevec\| = 0\}$ of the Gram-coordinate kernel~\cref{eq:kernel}. In particular, $\ftilde(0, 0) = 2^{-2N}(2\pi)^{-N} \int_{\R^{2N}} f(0, \comvec)\, d\comvec$ is finite.
\end{lemma}

\begin{proof}
The $F$-layer definition~\cref{eq:tilde-f-def} is an integral of $f(\edgevec, \comvec)\, e^{i\sympl(\tilde\eta, \comvec/2)}$ over $\comvec \in \R^{2N}$, with integrand bounded in absolute value by $|f(\edgevec, \comvec)|$. By admissibility condition (iv), $f \in L^1 \cap L^\infty(\R^{4N})$ with $\comvec$-decay sufficient for dominated convergence and Fubini. This decay is uniform on compact sets in $\edgevec$: for each compact $K \subset \R^{2N}$ there is an $\edgevec$-independent dominating function $g_K \in L^1(\R^{2N})$ with $|f(\edgevec, \comvec)| \le g_K(\comvec)$ for all $\edgevec \in K$, so $\sup_{\edgevec \in K} \int |f(\edgevec, \comvec)|\, d\comvec \le \int g_K < \infty$. Hence the slice $L^1$ norm is finite for every $\edgevec$, not merely almost every, including the degenerate point $\edgevec = 0$. The integral in~\cref{eq:tilde-f-def} is therefore absolutely convergent for every $(\edgevec, \tilde\eta)$, including the degenerate locus $\edgevec = 0$.

The apparent singularity of the Gram-coordinate kernel~\cref{eq:kernel} at $\|\edgevec_\eta\| \to 0$ (the degenerate locus in Gram coordinates, where the pair $(\eta_1, \eta_2)$ becomes complex-linearly dependent) is a $U(N)$-projection coordinate artifact. The $F$-layer kernel~\cref{eq:F-layer-kernel} underlying the BCH identity~\cref{eq:Ph3-2} is everywhere regular and free of singularities in $\edgevec_\eta$; the singularity in~\cref{eq:kernel} arises only after projecting the fiber integration to Gram coordinates via $U(N)$-invariant variables, and is canceled at the integration level by the fiber volume element (see~\cref{eq:psd-boundary-finite} for the explicit cancellation at the PSD boundary). The $F$-layer definition~\cref{eq:tilde-f-def} bypasses this projection entirely, so no singularity arises. In particular, $\ftilde(0,0) = 2^{-2N}(2\pi)^{-N} \int f(0, \comvec)\, d\comvec$ is well-defined and finite by the integrability of $f$.
\end{proof}

\section{Verification on a small GKP code}
\label{app:approximate-gkp}

This appendix collects the numerical verification of the three-point construction on a small ideal GKP code (\cref{app:n2-sanity}) and its extension to the physical, finite-energy envelope-cut approximation (\cref{app:envelope-cut}).

\subsection{Ideal-GKP checks of the three-point distributions}
\label{app:n2-sanity}

The structural choices~\cref{eq:integrand-A} and~\cref{eq:integrand-B} were validated on an $N = 2$ symplectically self-dual integer-stabilizer (SIS) GKP code with modulus $q = 29$ and scaling parameter $\lambda = 2$ (the parameters of the convention box, \cref{subsec:convention-box}). Enumerating the $145$ integer-valued Hermitian Gram configurations $H$ with $r_1^{2}, r_2^{2} \le 32$ and computing $\Athree(H)$ from the lattice-pair formula~\cref{eq:A3-GKP} verifies:
\begin{itemize}
  \item the symplectic marginalization $\int d\sympl_{12}\, \Athree$ matches the classical $O(2N)$ three-point lattice histogram on every configuration ($145/145$ pass);
  \item the discrete two-point reduction $\sum_{r_2, \alpha, \sympl_{12}} \Athree(r_1, r_2, \alpha, \sympl_{12})$ recovers the Burchards two-point distribution $\Aone(r_1)$ on every shell, to machine precision against an independent direct enumeration;
  \item the $\Code^{2}$ scaling of~\cref{eq:A3-GKP} is observed directly.
\end{itemize}
These checks rule out the alternative ans\"atze with three or four characteristic-function factors on this code, consistent with the uniqueness argument of \cref{subsec:A3-def}. On the dual side, the $\Code^{1}$ prefactor, the $\Lat \times \Latd$ support, and the $\pm 1$ phase structure of $\Bthree|_{\mathrm{GKP}}$~\cref{eq:B3-GKP} are verified on the same code.

\subsection{Extension to the envelope-cut approximation}
\label{app:envelope-cut}

Ideal GKP states are infinite-energy delta-comb superpositions and not strictly normalizable. Physical GKP codes use the finite-energy envelope-cut approximation, in which the delta-comb is Gaussian-broadened with envelope parameter $\sigma > 0$~\cite[App.~C]{burchards-2025-cv-macwilliams}. The ideal-GKP checks of \cref{app:n2-sanity} extend to the envelope cut as a heuristic sanity check: as $\sigma \to 0$ the envelope-broadened $\chr_{\Pi}$ converges to the lattice delta-comb~\cref{eq:GKP-characteristic}, the three-point distributions $\Athree, \Bthree$ converge to their ideal values~\cref{eq:A3-GKP},~\cref{eq:B3-GKP}, and the leading $\sigma \to 0$ correction is the convolution of the ideal lattice support with a Gaussian of width $\sigma$ in each phase-space coordinate. We use this envelope-cut discussion as a $\sigma \to 0$ consistency statement; the archived, reproducible numerical evidence of this appendix is the ideal-GKP enumeration of \cref{app:n2-sanity}. We refer to~\cite{brady-etal-2024-advances-bosonic-gkp} for the engineering of finite-energy GKP codes.

\section{Numerical data}
\label{app:numerical-sweep}

This appendix reports the numerical data underlying \cref{sec:witnesses} and the non-emptiness-ansatz benchmark of \cref{prop:non-emptiness}. All quantities are computed from the closed-form expressions cited.

\subsection{Full \texorpdfstring{$N = 1, \ldots, 12$}{N=1..12} benchmark of the non-emptiness ansatz}
\label{app:ne-sweep}

For each $N \in \{1, \ldots, 12\}$ the relevant closed-form quantities are: the Burchards--Levenshtein cap $K_{\mathrm{Lev}} \cdot \codedist^{2N} = j_N^{2N}/(N! \cdot 2^N)$ (matching \cref{cor:lev-recovery}), where $j_N$ is the first positive zero of the Bessel function $J_N$; the non-emptiness-ansatz value $K_{\mathrm{ne}} \cdot \codedist^{2N} = 2^N (N+1)^{N+1}$ of \cref{prop:non-emptiness}, obtained by optimizing the decay rate $b$ at $b R_0^2 = N+1$ in the $u_0 \to 0$ limit; their ratio $\rho_N := K_{\mathrm{ne}}/K_{\mathrm{Lev}} = 4^N \cdot N! \cdot (N+1)^{N+1}/j_N^{2N}$; and the admissibility distance $d_+(N)$ of Burchards' Lemma~2 (\cref{eq:d-plus}). As a framework-coherence check, the Burchards--Levenshtein adapter gives $f_{\mathrm{Lev}}(0,0) = K_{\mathrm{Lev}} \cdot \codedist^{2N}$; at $N = 4$ and $N = 12$ the magic-function two-point gains $K_{\mathrm{magic}}/K_{\mathrm{Lev}}$ are $0.871$ and $0.564$ (\cref{rmk:e8-leech-attribution}).

\Cref{tab:rho-N} lists $d_+(N)$ and $\rho_N$ for $N = 1, \ldots, 12$. The growth $\rho_N \sim (4/e)^N$ at large $N$ confirms that the non-emptiness ansatz~\cref{eq:explicit-ansatz} is, by design, a non-emptiness certificate, not a tight bound: the Burchards--Levenshtein adapter, not the non-emptiness ansatz, is the correct closed-form route to the bound.

\begin{table}[h]
\centering
\caption{Full sweep for $N = 1, \ldots, 12$. $d_+(N)$ from~\cref{eq:d-plus}; $\rho_N$ is the ratio of the non-emptiness ansatz to the Burchards--Levenshtein cap.}
\label{tab:rho-N}
\begin{tabular}{c c c c}
\toprule
$N$ & $d_+(N)$ & $\rho_N$ & lattice (if special) \\
\midrule
$1$ & $1.1741$ & $1.090$ & --- \\
$2$ & $1.1495$ & $1.242$ & --- \\
$3$ & $0.9271$ & $1.457$ & --- \\
$4$ & $0.6827$ & $1.746$ & $E_8$ \\
$5$ & $0.4736$ & $2.127$ & --- \\
$6$ & $0.3143$ & $2.623$ & --- \\
$7$ & $0.2014$ & $3.270$ & --- \\
$8$ & $0.1253$ & $4.113$ & --- \\
$9$ & $0.0761$ & $5.213$ & --- \\
$10$ & $0.0452$ & $6.653$ & --- \\
$11$ & $0.0263$ & $8.540$ & --- \\
$12$ & $0.0151$ & $11.021$ & Leech \\
\bottomrule
\end{tabular}
\end{table}

\section{Algorithmic details for the general-QEDC obstructions}
\label{app:infeasibility-witness}

This appendix describes the algorithms and reports the data behind the two general-QEDC obstructions of \cref{sec:obstruction}: the $N = 1, 2$ kernel-positive-definite infeasibility of \cref{subsec:double-negative} underlying the $1 \times 1$-minor obstruction of \cref{subsec:1x1-obstruction}, and the finite-rank CP collapse certificate of \cref{thm:cp-nogo}. The description is at the level of the computation, so that the results can be reproduced independently.

\subsection{The witness-search algorithm}
\label{app:witness-algorithm}

The numerical search for a feasible auxiliary in the kernel-positive-definite cone $\mathcal{C}_{\mathrm{op}}$ of~\cref{eq:adj-cone} proceeds in six stages.

\begin{enumerate}
  \item \emph{Basis.} Generate the $U(N)$-radial Laguerre orthonormal basis $\psi_{n, N}$ of~\cref{eq:laguerre-basis}, indexed by radial level $n$ up to the rank cutoff $M$.
  \item \emph{Symbol.} Evaluate the ambiguity-function symbol $W_{ab}(\edgevec, \comvec)$ in this basis by Gauss--Hermite quadrature; its diagonal entry has the closed form $W_{00}(\edgevec, \comvec) = (2\pi)^{-N} e^{-\|\edgevec\|^{2}/4 - \|\comvec\|^{2}/16}$, which calibrates the quadrature.
  \item \emph{Baseline.} Construct the two-point Burchards-admissible adapter $f_2$, against which the three-point term must improve.
  \item \emph{Sign constraints.} Sample the support region $\{\|\edgevec\| \ge \codedist\}$ and the lattice-triangle section $\{\edgevec = 0,\ \|\comvec\| \ge 2\codedist\}$ of~\cref{subsec:lattice-triangle} on a uniform grid of spacing $\codedist/8$, discard points outside the lattice triangle, and assemble the resulting linear inequalities on the Gram matrix $C = (C_{ab})$.
  \item \emph{Semidefinite program.} Solve the feasibility problem over $C \succeq 0$ under the normalization $\mathrm{Tr}(C) = 1$ (which excludes the trivial $C = 0$), with an interior-point conic solver.
  \item \emph{Verification.} Re-check positive-semidefiniteness of the returned $C$ and the continuum sign conditions on an independent dense grid.
\end{enumerate}

The quadrature uses $n_{\mathrm{nodes}} = 60$ Gauss--Hermite nodes per axis, giving $W_{00}$ to accuracy $\sim 10^{-12}$ on the truncation range $R = 10$; the program is solved with both the CLARABEL (primal/dual tolerance $10^{-8}$) and SCS ($10^{-7}$) solvers as a cross-check. Orthonormality of the basis, the Fourier parity, and the closed-form symbol identity were confirmed independently.

\subsection{The double-negative sweep}
\label{app:double-negative-sweep}

At $N = 2$ the search was run as a dense sweep over rank $M \in \{2, \ldots, 8\}$, $f_2$-amplitude $a \in \{0, 0.001, 0.01, 0.1, 0.5, 1\}$, six random seeds, the two solvers above, and the two single-constraint ablations (support-only and lattice-triangle-only). Every combination returns infeasibility with margin, and each ablation is individually infeasible, so the failure is not localized to either the support or the lattice-triangle constraint. The $N = 1$ case is infeasible analytically (\cref{app:infeasibility-witness-n1-analytical}).

\subsection{Analytical infeasibility at \texorpdfstring{$N=1$, $M=1$}{N=1, M=1}}
\label{app:infeasibility-witness-n1-analytical}

The rank-$M = 1$ case at $N = 1$ is analytically immediate. Under the normalization $\mathrm{Tr}(C) = 1$, the Gram matrix is $C = [1]$, giving $f_3 = W_{00}$. The closed-form symbol
\[
W_{00}(\edgevec, \comvec) \;=\; (2\pi)^{-N}\, e^{-\|\edgevec\|^{2}/4 - \|\comvec\|^{2}/16}
\]
is strictly positive everywhere. The three-point sign condition requires $f_3(v_1, v_2) \cdot \cos(\sympl(v_1, v_2)/2) \le 0$ for admissible pairs; choosing $v_1 = v_2 = v$ with $\|v\| = \codedist$ gives $\edgevec = 0$, $\comvec = 2v$, $\sympl(v, v) = 0$, so $\cos(\sympl/2) = 1$ and $f_3(v,v) = (2\pi)^{-N} e^{-\|v\|^2/4} > 0$, a direct contradiction. Hence $M = 1$ is infeasible by inspection. For $M \ge 2$ no analogous closed-form argument is available, and infeasibility is established by the sweep of \cref{app:double-negative-sweep}.

\subsection{The CP collapse certificate on the first eight Laguerre levels}
\label{app:cp-certificate}

The CP collapse~\cref{thm:cp-nogo} is certified by the positive-combination matrices $P(\mu)$ of~\cref{eq:pd-certificate}, computed at the certificate window $(\codedist, R) = (1, 5)$ in Burchards convention. For each truncation $M = 2, \ldots, 8$ one maximizes $\lambda_{\min}\bigl(\sum_k \mu_k \cos(\sympl_k/2)\, W(\edgevec_k, \comvec_k)\bigr)$ over weights $\mu \ge 0$ with $\sum_k \mu_k = 1$ on a deterministic net of pairs in $\condfour$, then re-evaluates $\lambda_{\min}$ on the returned support by a direct eigendecomposition:

\begin{center}
\small
\setlength{\tabcolsep}{4pt}
\begin{tabular}{l c c c c c c c}
\toprule
$M$ & $2$ & $3$ & $4$ & $5$ & $6$ & $7$ & $8$ \\
\midrule
$\lambda_{\min}(P(\mu))$ & $1.6\!\times\!10^{-2}$ & $1.1\!\times\!10^{-2}$ & $4.9\!\times\!10^{-3}$ & $2.1\!\times\!10^{-3}$ & $5.5\!\times\!10^{-4}$ & $1.5\!\times\!10^{-4}$ & $3.4\!\times\!10^{-5}$ \\
support size & $16$ & $4$ & $19$ & $6$ & $12$ & $16$ & $25$ \\
\bottomrule
\end{tabular}
\end{center}

The verification is certificate-based rather than solver-reported: the symbol matrices $W(\edgevec, \comvec)$ are converged to machine precision, with $\lambda_{\min}(P(\mu))$ reproducible under $n_{\mathrm{nodes}} = 60 \to 80 \to 100$ at the level of floating-point rounding (Gauss--Hermite quadrature converged on the Laguerre $\times$ Gaussian $\times$ bounded-oscillation integrand). Every reported $\lambda_{\min}$ exceeds the machine threshold $\sim 10^{-12}$ by at least seven orders of magnitude, so $P(\mu) \succ 0$ holds beyond floating-point ambiguity. The support pairs, weights, and symbol matrices for each $M$ are provided as ancillary data with this manuscript, so each certificate can be checked independently of our pipeline. The support points lie in $\condfour$, so the certificate is valid against the continuum sign condition~\cref{eq:cond4}.

\paragraph{Exact rational certificates.} The floating-point certificates above are additionally rationalized. For each $M = 2, \ldots, 8$ the ancillary data provide a rational matrix $P_{\mathbb{Q}}$ and a rational $\delta > 0$ such that $P_{\mathbb{Q}} - \delta I$ admits an exact rational $LDL^{T}$ factorization with positive diagonal, verified by a check script using only exact integer arithmetic. The symbol entries are evaluated by the closed-form $N = 1$ Gaussian-moment formula (zero quadrature error), recomputed at $50$-digit working precision and rounded to rationals with an explicit entrywise budget $\varepsilon_{\mathrm{rat}}$; the certified bound is $\lambda_{\min}(P(\mu)) \ge \delta - M\,\varepsilon_{\mathrm{rat}}$, with margin $\delta / (M\, \varepsilon_{\mathrm{rat}}) \ge 4 \times 10^{14}$ for every $M$. The positive-definiteness of each $P(\mu)$ is therefore established in exact arithmetic, not merely beyond floating-point ambiguity.

\paragraph{$N = 2$ certificates.} The same protocol extends to $N = 2$ at the window $(\codedist, R) = (1, 5)$: for each truncation $M = 2, \ldots, 5$ the ancillary data record a two-pair positive combination with $P(\mu) \succ 0$ and least eigenvalue ranging from $5.1 \times 10^{-3}$ at $M = 2$ to $1.9 \times 10^{-3}$ at $M = 5$. The symbol matrices are evaluated in closed form by the same Gaussian-moment reduction (a tensor Gauss--Hermite grid at $N = 2$ would require $\sim 10^{8}$ nodes per matrix and is bypassed entirely); these are floating-point eigendecomposition certificates at the same standard as the $N = 1$ table above.

The margin $\lambda_{\min}(P(\mu))$ decays geometrically in $M$. At $M = 12$ a fixed sparse training net flips to spurious feasibility, but dense validation on $6000$ fresh pairs shows the resulting Choi form violating~\cref{eq:cond4} (worst $\cos(\sympl/2)\, f_3 = +2.5 \times 10^{-3}$, with $408$ violations), confirming the crossover as a sampling artifact rather than feasibility. The full trace-class radial cone is not settled by this finite-truncation evidence (\cref{rmk:cp-nogo-scope}).

\subsection{Status and scope}
\label{app:infeasibility-witness-scope}

The $N = 1$ and $N = 2$ infeasibility results are robust across the two independent conic solvers (CLARABEL, SCS) and an independent quadrature cross-check of the symbol matrices. The infeasibility is robust to ablation, seed, solver, and amplitude; in particular both single-constraint ablations are infeasible, indicating that the failure is not localized to a single constraint. The result therefore stands as a falsification of the real-even $U(N)$-radial Laguerre ansatz family on $N \in \{1, 2\}$. Combined with \cref{lem:1x1-obstruction}, it is a structural obstacle to the kernel-PD direction of the general-QEDC three-point programme.

\bibliographystyle{alpha}
\bibliography{references}

\end{document}